\documentclass[dvipsnames,11pt]{article}
\usepackage{amsmath,amssymb,amsthm,color,bm,mathrsfs,extarrows,tikz,graphicx,mathtools,enumitem,fancybox,makecell}
\usepackage[square,sort,comma,numbers]{natbib}
\usepackage{subcaption}
\usepackage[ruled]{algorithm2e}
\usepackage[margin=1in]{geometry}
\usepackage{xcolor,bbm}
\usepackage[pagebackref]{hyperref}
\usepackage{comment}
\usepackage{soul}
\usepackage{thm-restate}
\usepackage{ifthen}
\usepackage{mathdots}
\usepackage{enumitem}
\usepackage{float}
\usepackage{quantikz}

\allowdisplaybreaks

\setlist[itemize]{itemsep=0pt}
\setlist[enumerate]{itemsep=0pt}
\hypersetup{colorlinks=true,urlcolor=blue,linkcolor=blue,citecolor=[rgb]{.42,.56,.14},}
\usetikzlibrary{decorations.pathreplacing,decorations.markings,decorations.pathmorphing,decorations.shapes,arrows.meta,positioning,patterns, quantikz2}

\usepackage[capitalise,nameinlink]{cleveref}
\Crefname{lemma}{Lemma}{Lemmas}
\Crefname{fact}{Fact}{Facts}
\Crefname{theorem}{Theorem}{Theorems}
\Crefname{corollary}{Corollary}{Corollaries}
\Crefname{claim}{Claim}{Claims}
\Crefname{example}{Example}{Examples}
\Crefname{problem}{Problem}{Problems}
\Crefname{definition}{Definition}{Definitions}
\Crefname{notation}{Notation}{Notations}
\Crefname{assumption}{Assumption}{Assumptions}
\Crefname{subsection}{Subsection}{Subsections}
\Crefname{section}{Section}{Sections}
\Crefformat{equation}{(#2#1#3)}

\newtheorem{theorem}{Theorem}[section]
\newtheorem*{theorem*}{Theorem}

\newtheorem{proposition}[theorem]{Proposition}
\newtheorem*{proposition*}{Proposition}

\newtheorem*{property*}{Property}
\newtheorem{lemma}[theorem]{Lemma}
\newtheorem*{lemma*}{Lemma}
\newtheorem{corollary}[theorem]{Corollary}
\newtheorem*{corollary*}{Corollary}
\newtheorem*{conjecture*}{Conjecture}
\newtheorem{fact}[theorem]{Fact}
\newtheorem*{fact*}{Fact}

\newtheorem*{exercise*}{Exercise}

\newtheorem*{hypothesis*}{Hypothesis}

\theoremstyle{definition}
\newtheorem{definition}[theorem]{Definition}

\newtheorem{exercise-easy}[theorem]{Exercise}
\newtheorem{exercise-med}[theorem]{Exercise}
\newtheorem{exercise-hard}[theorem]{Exercise$^\star$}
\newtheorem{claim}[theorem]{Claim}
\newtheorem*{claim*}{Claim}

\newtheorem{remark}[theorem]{Remark}
\newtheorem*{remark*}{Remark}

\newtheorem*{observation*}{Observation}

\DeclareSymbolFont{extraup}{U}{zavm}{m}{n}
\DeclareMathSymbol{\varheart}{\mathalpha}{extraup}{86}
\DeclareMathSymbol{\vardiamond}{\mathalpha}{extraup}{87}

\DeclareMathOperator*{\E}{\mathbb E}

\renewcommand{\Pr}{\operatorname*{\mathbf{Pr}}}

\newcommand{\Mod}[1]{\ (\mathrm{mod}\ #1)}

\newcommand{\eps}{\varepsilon}
\newcommand{\abs}[1]{\left| #1 \right|}
\newcommand{\vabs}[1]{\left\| #1 \right\|}
\newcommand{\abra}[1]{\left\langle #1 \right\rangle}
\newcommand{\pbra}[1]{\left( #1 \right)}
\newcommand{\sbra}[1]{\left[ #1 \right]}
\newcommand{\cbra}[1]{\left\{ #1 \right\}}
\newcommand{\floorbra}[1]{\left\lfloor #1 \right\rfloor}

\renewcommand{\mid}{\,\middle\vert\,}

\newcommand{\bin}{\{0,1\}}

\newcommand{\tow}{\mathrm{tow}}
\newcommand{\ac}{\mathsf{AC^0}}
\newcommand{\nc}{\mathsf{NC^0}}
\newcommand{\nicksclass}{\mathsf{NC}}
\newcommand{\qnc}{\mathsf{QNC^0}}
\newcommand{\qac}{\mathsf{QAC^0}}
\newcommand{\qpoly}{\mathsf{qpoly}}

\newcommand{\Dhard}{\mathcal{D}_{\mathsf{hard}}}
\newcommand{\Dhost}{\mathcal{D}_{\mathsf{host}}}
\newcommand{\reduc}{\mathsf{red}}
\newcommand{\complexi}{\mathsf{i}}

\newcommand{\CNOT}{\mathsf{CNOT}}

\newcommand{\Tof}{\mathsf{Tof}}

\newcommand{\CS}{\mathsf{CS}}
\newcommand{\Had}{\mathsf{H}}

\newcommand{\Cbb}{\mathbb{C}}

\newcommand{\Nbb}{\mathbb{N}}
\newcommand{\Rbb}{\mathbb{R}}
\newcommand{\Zbb}{\mathbb{Z}}

\newcommand{\Ccal}{\mathcal{C}}
\newcommand{\Dcal}{\mathcal{D}}
\newcommand{\Ecal}{\mathcal{E}}

\newcommand{\Pcal}{\mathcal{P}}
\newcommand{\Qcal}{\mathcal{Q}}

\newcommand{\Tcal}{\mathcal{T}}
\newcommand{\Ucal}{\mathcal{U}}

\newcommand{\Wcal}{\mathcal{W}}

\newcommand{\tvdist}[1]{\vabs{#1}_\mathsf{TV}}

\renewcommand{\tilde}{\widetilde}

\title{Quantum Advantage from Sampling Shallow Circuits: Beyond Hardness of Marginals}

\author{
Daniel Grier\thanks{UC San Diego. Email: \texttt{dgrier@ucsd.edu}.}
\and
Daniel M. Kane\thanks{UC San Diego. Email: \texttt{dakane@ucsd.edu}. Supported in part by NSF Medium Award CCF-2107079.}
\and
Jackson Morris\thanks{UC San Diego. Email: \texttt{jrm035@ucsd.edu}.}
\and
Anthony Ostuni\thanks{UC San Diego. Email: \texttt{aostuni@ucsd.edu}.}
\and
Kewen Wu\thanks{Institute for Advanced Study. Email: \texttt{shlw\_kevin@hotmail.com}. Supported by the National Science Foundation under Grant No. DMS-2424441, and by the IAS School of Mathematics.}
}
\date{}

\begin{document}

\maketitle

\begin{abstract}
    We construct a family of distributions $\cbra{\mathcal{D}_n}_n$ with $\mathcal{D}_n$ over $\bin^n$ and a family of depth-$7$ quantum circuits $\cbra{C_n}_n$ such that $\mathcal{D}_n$ is produced exactly by $C_n$ with the all zeros state as input, yet any constant-depth classical circuit with bounded fan-in gates evaluated on any binary product distribution has total variation distance $1 - e^{-\Omega(n)}$ from $\mathcal{D}_n$.
    Moreover, the quantum circuits we construct are geometrically local and use a relatively standard gate set: Hadamard, controlled-phase, CNOT, and Toffoli gates.
    All previous separations of this type suffer from some undesirable constraint on the classical circuit model or the quantum circuits witnessing the separation.
    
    Our family of distributions is inspired by the Parity Halving Problem of Watts, Kothari, Schaeffer, and Tal (STOC, 2019), which built on the work of Bravyi, Gosset, and K\"onig (Science, 2018) to separate shallow quantum and classical circuits for relational problems.
\end{abstract}

\section{Introduction}

One of, if not \emph{the} primary direction in the study of quantum computing is to exhibit computational tasks that can be performed far more efficiently on a quantum computer than on a classical one.
There are a number of promising candidates \cite{shor1999polynomial, aaronson2013computational, bouland2019complexity}, but the quantum superiority of many such algorithms relies on unproven assumptions about computational hardness.

To obtain \emph{unconditional} quantum-classical separations, one must consider classical models of computation against which there are known unconditional lower bounds. Bravyi, Gosset, and K\"onig gave the first result of this kind by constructing a search problem which could be solved by constant-depth quantum circuits, but not constant-depth classical circuits \cite{bravyi2018quantum}. Formally, $\mathsf{FQNC^0} \not\subseteq \mathsf{FNC}^0$. Since then, there have been many improvements to this result that consider stronger classical circuit families, different error models, and/or different topologies \cite{watts2019exponential, grier2020interactive, bravyi2020quantum, hasegawa2021quantum, caha2023colossal}.

Nevertheless, these results still fundamentally use the search paradigm for separating the quantum and classical circuit models (or, in fact, sometimes generalizations of search \cite{grier2020interactive,grier2021interactive}). Intuitively, one can think of these search problems as follows: the input to the problem is both the number of qubits $n$ \emph{and} a specification of a constant-depth quantum circuit $Q$, and the goal is to output any bit string in the support of the distribution after measuring $Q \ket{0^n}$ in the computational basis. One might wonder if the specification of the quantum circuit is even necessary. That is, is there a single quantum circuit for every $n$ that gives rise to a hard-to-sample distribution? In fact, Bravyi, Gosset, and K\"onig asked exactly this question in their original work \cite[Section 5]{bravyi2018quantum}.

There are a few reasons why we might want such a separation. First, one goal for proofs of quantum advantage is to help distill the core aspects of quantum computers that make them more powerful than their classical counterparts. Clearly then, a separation from a single family of distributions is desirable in its simplicity. Moreover, such results give complexity-theoretic support for certain quantum advantage experiments in which changing the underlying circuit is extremely difficult \cite{zhong2020quantum, deng2023gaussian, young2024atomic}.

Watts and Parham \cite{watts2023unconditional} were the first to answer the challenge of \cite{bravyi2018quantum} by constructing a family of constant-depth quantum circuits with output distributions that cannot be sampled (even approximately) by constant-depth classical circuits with bounded fan-in gates. Unfortunately, their result has two significant caveats. First, it imposes a strict requirement on the number of input bits to the classical circuit.  Second, the quantum circuits they construct contain more-or-less arbitrary single-qubit gates (at least outside the Clifford hierarchy). 

These two properties combine to make the ``quantum'' contribution to the quantum-classical separation less clear. To see this, first notice that the usual method of converting between gate sets does not apply in the constant-depth regime, since the Solovay-Kitaev theorem incurs a polylogarithmic depth overhead \cite{kuperberg2023breaking}. This implies that the choice of which single-qubit gates to allow in the quantum circuit model could ultimately affect which kinds of separations are possible. This consideration has been put into sharp relief by recent work that gives a \emph{product} distribution which cannot be sampled (even approximately) by constant-depth classical circuits with uniformly random input bits \cite{viola2023new, kane2024locality}. 

In other words, it is possible to obtain a quantum vs.\ classical separation with a quantum circuit model that has no entangling gates (as only single-qubit rotation gates are needed to sample from a product distribution), undermining the claim that the separation is related to the powers of quantum mechanics. Indeed, if instead we allowed our classical circuit to have random inputs of arbitrary bias, then they, too, could easily produce the desired distribution. While the result of \cite{watts2023unconditional} does allow for classical circuits with biased input bits, the restriction on the number of input bits leaves open the possibility that larger classical circuits may still be able to sample from the distribution.

The main contribution of this work is the construction of a family of distributions that achieves the best properties from all prior works: 

\begin{theorem}[{Informal Version of \Cref{thm:main_quantitative}}]\label{thm:main}
There is a uniform family of constant-depth quantum circuits $\{Q_n\}_n$ such that
\begin{itemize}
\item \emph{Discrete gate set:} $Q_n$ is constructed from Hadamard, controlled-phase, $\CNOT$, and Toffoli gates. Furthermore, $Q_n$ has a depth-7, geometrically local implementation.
\item \emph{Quantum advantage:} Let $\{C_n \colon \{0,1\}^* \to \{0,1\}^n\}_n$ be a family of constant-depth classical circuits (i.e., $\nc$), and consider the following two distributions: $C_n$ applied to any product distribution; and measuring $Q_n \ket{0^n}$ in the computational basis. The total variation distance between these two distributions is $1-e^{-\Omega(n)}$.
\end{itemize}
\end{theorem}

Our distribution in \Cref{thm:main} is based on a relational problem given by Watts, Kothari, Schaeffer, and Tal \cite{watts2019exponential} to strengthen the previously mentioned separation between $\nc$ and $\qnc$ \cite{bravyi2018quantum}.
Despite this similar construction, our results are incomparable, as they lower bound the stronger class of $\ac$, but our results are in the distributional, rather than relational, setting.

\Cref{thm:main} may also be of modest philosophical interest.
Recall that randomness extractors convert weak sources of randomness into a near uniform distribution.
In influential work, Trevisan and Vadhan provided extractors for distributions produced by polynomial size circuits, claiming that these ``sampl[e]able distributions are a reasonable model for distributions actually arising in nature'' \cite{trevisan2000extracting}.
A recent follow up by Ball, Goldin, Dachman-Soled, and Mutreja instead argues that a better choice for ``natural sources'' are those generated by quantum circuits, since the universe is governed by quantum phenomena \cite{ball2023extracting}.
Our main result shows that even in the extremely restricted circuit regime, these two beliefs dramatically differ.

\paragraph{Open Problems.}
The obvious next question in this line of inquiry, raised earlier in \cite[Section 2]{watts2023unconditional}, is whether a similar distributional separation exists between the classes of $\ac$ and $\qnc$.
In fact, it appears even the weaker task of separating $\ac$ from $\qac$ (for sampling distributions) is open.
There are several known distributions based on pseudorandom objects that cannot be accurately sampled in $\ac$ \cite{lovett2011bounded, beck2012large, viola2014extractors, viola2020sampling}; it is unclear whether shallow quantum circuits can sample them.
We remark that while our distribution is based on a problem from \cite{watts2019exponential} which separates $\ac$ and $\qnc$ for relational problems, our distribution \emph{can} be sampled by an $\ac$ circuit (see \Cref{rmk:improved_upper}).

Another direction is to refine the quantum gate set.
The quantum circuits in our main separation result only require Hadamard, controlled-phase, and Toffoli gates, as opposed to the rotation gates required to generate the $(1/3)$-biased product distribution used in previous separations \cite{viola2023new, kane2024locality}.
Still, one may wish to further limit the gate set, especially in light of the fact that Hadamard and Toffoli gates are quantum universal \cite{aharonov2003simple}.
Unfortunately, the standard techniques to simulate the controlled-phase gates in this manner do not naively work in our setting (see \Cref{sec:optimal_gateset}), and we leave the minimal gate set required to separate $\nc$ and $\qnc$ for sampling distributions as an open question.

One final direction deserving of further investigation is hardness amplification for sampling.
Our proof of \Cref{thm:main} crucially uses a direct product theorem for sampling in $\nc$ (see \Cref{sec:hardness_amp}), which allows us to amplify a weak separation between $\nc$ and $\qnc$.
A similar direct product theorem (or more generally, a hardness amplification result) for sampling in $\ac$ would likely be useful in addressing open separations.
Note that such a theorem was asked for by Chattopadhyay, Goodman, and Zuckerman \cite{chattopadhyay2022space} who gave an analogous result for read-once branching programs.

\paragraph{Paper Overview.}
We provide an overview of the proof of a precise version of \Cref{thm:main} in \Cref{sec:outline}.
\Cref{sec:prelims} contains background material and several useful results applied in later sections.
The quantum sampleability of \Cref{thm:main} is given in \Cref{sec:qnc0}, while the classical hardness of \Cref{thm:main} is in \Cref{sec:nc0}.
\Cref{sec:nc0} also contains various sampling schemes for related distributions and a direct product theorem for sampling in $\nc$.

\section{The Proof Outline}\label{sec:outline}

In this section, we provide an overview of the proof of \Cref{thm:main_quantitative} -- a more precise and quantitative version of \Cref{thm:main} parameterized by \emph{locality}.
A function $f\colon \bin^* \to \bin^n$ is $d$-local if no output bit depends on more than $d$ input bits.
Note that any family of $\nc$ circuits computes functions of constant locality.
We will often refer directly to a distribution as $d$-local if it is the result of applying a $d$-local function to random inputs drawn from a product distribution.

\begin{theorem}\label{thm:main_quantitative}
    Let $d \ge 1$ be an integer.
    There exists a constant $c_d > 0$ depending only on $d$ such that the following holds.
    There is a uniform family of distributions $\{\Dcal_n\}_{n \ge c_d}$ with $\Dcal_n$ over $\bin^{n}$ such that 
    \begin{itemize}
    \item There exists a family of geometrically local depth-$7$ quantum circuits $\{C_n\}_{n \ge c_d}$ where $\Dcal_n$ is produced exactly by $C_n$ on input $\ket{0^{2n}}$.
    In addition, the quantum circuits only uses Hadamard, controlled-phase, $\CNOT$, and Toffoli gates, and measurements in the computational basis. 
    Moreover, Hadamard gates are only applied in the first and last layers, i.e., $\{C_n\}_n$ is in the second level of the Fourier Hierarchy \cite{shi2005quantum}.
    \item For all $n \ge c_d$, $\Dcal_n$ has total variation distance $1- e^{-n/c_d}$ from any $d$-local distribution with any binary product distribution as input.
    \end{itemize}
\end{theorem}

\begin{remark}
    Given \Cref{thm:main_quantitative}, it is natural to wonder whether every distribution produced by $\nc$ circuits can be sampled by $\qnc$ circuits.
    The following example shows that this is not the case.
    Consider the distribution $\Pcal$ over $\bin^n$ which takes value $0^n$ with probability $1/2$ and $1^n$ otherwise.
    A classical circuit can easily produce $\Pcal$ by having each output bit mirror the same input bit.
    $\qnc$ circuits, however, cannot generate $\Pcal$, as doing so is equivalent to preparing a nekomata (first defined in \cite{rosenthal2021bounds}), i.e., a state of the form 
    \begin{align*}
        \ket{\psi} &= \frac{\ket{0^n}\ket{\psi_0} + \ket{1^n}\ket{\psi_1}}{\sqrt{2}}.
    \end{align*}
    A lightcone argument shows that $\Omega(\log{n})$ depth is necessary to prepare such a state, as is shown in \cite{watts2019exponential}. 
\end{remark}

In \Cref{sub:q_sampling}, we will review the Parity Halving Problem of \cite{watts2019exponential}, and explain how to derive a distributional version of the problem that can be exactly sampled by a shallow quantum circuit, but seemingly cannot be accurately sampled by a function of low locality.
In \Cref{sub:overview_hardness}, we will sketch a proof that this distribution has constant distance from every $d$-local distribution.
That is, there is a distribution which exhibits a constant distance separation between classical and quantum sampling with shallow circuits.
To boost this separation to an optimal one, we highlight and apply a direct product theorem implicit in \cite{kane2024locality} in \Cref{sub:overview_boosting}.

\subsection{Quantum Sampling and a Classical Reduction}\label{sub:q_sampling}

As mentioned, the authors of \cite{watts2019exponential} define the \emph{Parity Halving Problem} (PHP): a relation problem over bit strings which is solvable by a shallow quantum circuit, but any randomized $\ac$ (and therefore $\nc$) circuit can only succeed on a trivial fraction of inputs. It is defined as follows:
\begin{definition}[Parity Halving Problem]\label{def:PHP}
    Given $x \in \{0, 1\}^n$ with $|x| \equiv 0 \Mod{2}$, return $y \in \{0, 1\}^{n}$ which satisfies $|y| \equiv \frac{|x|}{2} \Mod{2}$.
\end{definition}
Their initial observation is that the PHP can be solved with certainty on all instances by a $\qnc$ circuit with polynomial size quantum advice, i.e., PHP is in the class $\qnc/\qpoly$. This circuit is shown on the left in \Cref{fig:php_circuit}.

\begin{figure}[ht]
 \centering
 \input{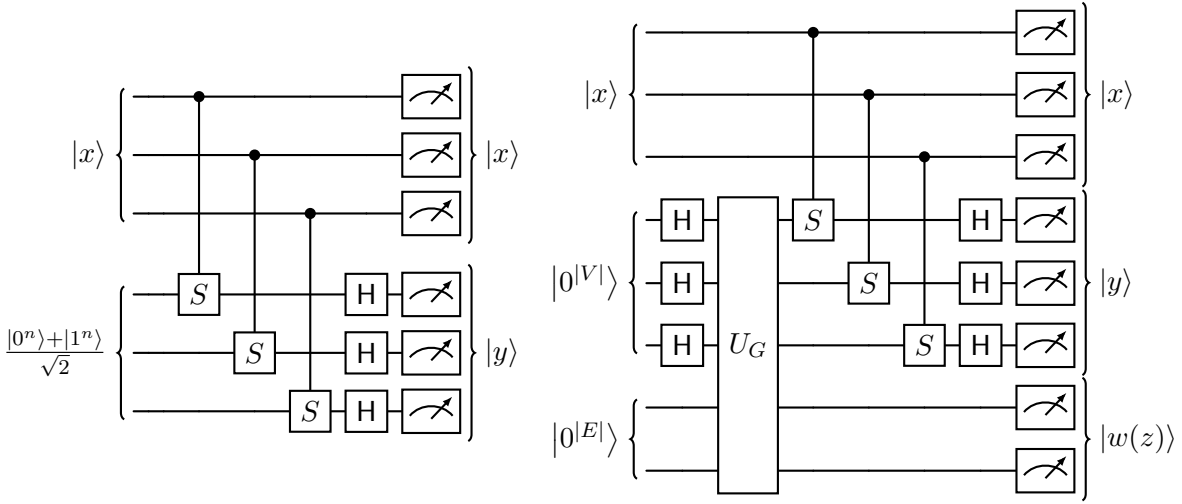}
 \caption{On the left is a $\qnc/\qpoly$ circuit which solves the Parity Halving Problem and on the right is a $\qnc$ circuit which solves the Relaxed Parity Having Problem over a graph $G = (V, E)$. Here $U_G$ is the ($|V| + |E|$)-qubit unitary which acts as $U_G\ket{z}\ket{b} = \ket{z}\bigotimes_{e = (u, v) \in E}\ket{b_e\oplus z_u \oplus z_v}$ for all $z \in \{0, 1\}^V$ and $b \in \{0, 1\}^E$.}
 \label{fig:php_circuit}
\end{figure}

To see that this circuit does indeed solve the PHP, note that after the $\CS$ gates are applied the resulting state is
\begin{align*}
    \ket{x}\otimes \frac{\ket{0^n} + \complexi^{x_1 + \cdots + x_n}\ket{1^n}}{\sqrt{2}} &= \ket{x}\otimes \frac{\ket{0^n} + \complexi^{|x|}\ket{1^n}}{\sqrt{2}}
    = \begin{cases}
        \ket{x}\otimes \frac{\ket{0^n} + \ket{1^n}}{\sqrt{2}} & \text{ if } |x| \equiv 0\Mod{4},\\
        \ket{x}\otimes \frac{\ket{0^n} - \ket{1^n}}{\sqrt{2}} & \text{ if } |x| \equiv 2\Mod{4}.
    \end{cases}
\end{align*}
Finally, applying $\Had^{\otimes n}$ to $\frac{\ket{0^n} + (-1)^b\ket{1^n}}{\sqrt{2}}$ yields a uniform superposition over bit strings of parity $b$.

In order to obtain a relational separation between $\nc$ and $\qnc$ without the need for quantum advice\footnote{Actually, the Relaxed Parity Halving Problem even serves to separate $\qnc$ and $\ac$, but we make mention of it here as it serves as motivation for our sampling separation.}, the authors of \cite{watts2019exponential} define a variant of the PHP as follows:
\begin{definition}[Relaxed Parity Halving Problem]\label{def:RPHP}
    Fix a graph $G = (V, E)$.
    Given $x \in \{0, 1\}^{V}$ with $|x| \equiv 0 \Mod{2}$, return $y \in \{0, 1\}^{V}$ and $w \in \{0, 1\}^{E}$ for which there exists $z \in \{0, 1\}^V$ such that
    \[
        z_u \oplus z_v = w_{(u, v)} \ \ \forall (u, v) \in E \quad\text{and}\quad |y| \equiv \langle z, x \rangle + \frac{|x|}{2} \Mod{2}.
    \]
\end{definition}
 If $G$ has a cycle then it may not be the case that for each $w \in \{0, 1\}^E$ there exists a $z$ such that $w$ and $z$ together satisfy the first condition above.
 However, when the underlying graph $G$ is a tree then such a $z$ exists for each $w$ and preparing a ``poor man's cat state'' suffices to solve the Relaxed Parity Halving Problem over $G$. A poor man's cat state is a state proportional to $\ket{z} + \ket{\overline z}$ where $\overline z$ is the bitwise negation of $z$.
The key observation is that there \emph{is} a $\qnc$ circuit which prepares a poor man's cat state, conditioned on the measurement outcome of another register. Indeed, the state $\frac{1}{\sqrt{2^{|V| - 1}}}\sum_{z \in \{0, 1\}^n, z_1 = 0} \ket{w(z)}\otimes\frac{\ket{z} + \ket{\overline{z}}}{\sqrt{2}}$ (where $w(z)$ and $z$ satisfy the first constraint of \Cref{def:RPHP}) can be prepared by a $\qnc$ circuit so long as the maximum degree of $G$ is constant. Finally, applying the PHP circuit, treating the $Z$ register of the poor man's cat state as if it were the cat state, yields a uniformly random string of parity $\frac{|x|}{2} + \langle x, z \rangle$. 

This circuit is shown on the right in \Cref{fig:php_circuit}. In order to obtain lower bounds against $\nc$ for the (R)PHP, standard locality arguments apply.

Recall our goal is to construct a \emph{distributional} separation.
A reasonable first attempt might be to consider the distribution $\Dcal_{\text{RPHP}}$ which is uniform over tuples $(x, y, w)$ satisfying the relation.
If generating this distribution is as hard as computing the RPHP relation, then classical hardness follows. 
Unfortunately, this is not the case.
To gain some intuition, consider the classical \textsf{PARITY} function, which cannot be computed by shallow classical circuits \cite{haastad1986computational, smolensky1987algebraic}, and yet, a simple $\nc$ circuit can sample from the distribution $(X, \textsf{PARITY}(X))$ where $X$ is a uniformly random bit string \cite{babai1987random, boppana1987one}.
Specifically, one can map the random bits
\[
    y_1, y_2, \dots, y_n \to ((y_1 \oplus y_2) \circ (y_2 \oplus y_3) \circ \cdots \circ (y_{n-1} \oplus y_n), y_1 \oplus y_n),
\]
where $\circ$ denotes concatenation \cite{babai1987random, boppana1987one}.
In fact, a similar construction classically samples from $\Dcal_{\text{RPHP}}$ (see \Cref{sec:nc0_upper}).

We briefly digress to remark that this example is not an outlier.
Indeed, the past decade or two has seen the study of sampling distributions blossom into a rich area, in many ways independent of computation, with exciting connections to fields such data structures \cite{viola2012complexity, lovett2011bounded, beck2012large, viola2023new, yu2024sampling, kane2024locality}, extractors \cite{trevisan2000extracting, de2012extractors, viola2012extractors, viola2014extractors, ball2023extracting}, pseudorandom generators \cite{viola2012complexity, lovett2011bounded}, and coding theory \cite{shaltiel2024explicit}.
We refer the interested reader to the recent works \cite{filmus2023sampling, viola2023new, kane2024locality, yu2024sampling, shaltiel2024explicit, kane2024locally} and references within for more details. 

To overcome the above barrier, consider the strings $(x, y, w)$ subject to the constraint $|x| \equiv 1 \Mod{2}$. 
The simple-but-key observation is that on input $x$ with odd Hamming weight, the quantum circuit shown on the left in \Cref{fig:php_circuit} yields a uniformly random bit string $y$. (Note that $w$ is always uniformly random.) Hence, if we replace $\ket{x}$ with some other state, we can now run our quantum circuit without necessarily invoking the promise on the Hamming weight of $x$, which gives us some added flexibility in our choice of distribution. In fact, we will simply replace each qubit with the single-qubit state $\sqrt{3/4}\ket{0} + \sqrt{1/4}\ket{1}$. That is, if we were to measure the qubits of the $x$-register, we would obtain the $(1/4)$-biased distribution for each bit. It is exactly this distribution of inputs for which we can show classical hardness.
In the following subsection, we will highlight exactly \emph{why} this distribution does not suffer from the same shortcoming as the \textsf{PARITY} example.

Formally, the distribution witnessing the separation is defined as follows:
\begin{definition}[The $\Dhost(\Tcal)$ Distribution]\label{def:dhost}
Let $\Tcal=(V,E)$ be a tree with undirected edges.
A sample $(X,Y,W)\sim\Dhost(\Tcal)$ is drawn as follows:
first sample $X\sim\Ucal_{1/4}^V$ and $Z\sim\Ucal_{1/2}^V$. 
Define $W\in\bin^E$ by setting $W_e=Z_u\oplus Z_v$ for each $e=\{u,v\}\in E$.
If $X$ has odd Hamming weight, then sample $Y\in\bin^V$ uniformly at random; otherwise, sample $Y$ as a uniform $|V|$-bit string of parity $\abra{Z,X}+|X|/2\Mod2$.
\end{definition}
In \Cref{sec:qnc0}, we show how to sample $\Dhost(\Tcal)$ exactly using $\qnc$ circuits with the help of ancilla. The circuit is obtained by slightly modifying the construction given in \cite{watts2019exponential} for the RPHP:
\begin{restatable}{proposition}{propqnczero}\label{prop:qnc0}
Let $\Tcal=(V,E)$ be a tree and let $\Delta\ge2$ be its maximum vertex degree.
Then there exists a geometrically local quantum circuit $C$ such that the following holds.
\begin{itemize}
\item $C$ has depth $2\Delta + 1$ and only uses Hadamard, controlled-phase, CNOT, and Toffoli gates. Moreover, Hadamard gates are only applied in the first and last layers.
\item Let $\Pcal$ be the distribution obtained by measuring $C\ket{0^{5|V|-1}}$ in the computational basis. Then the marginal distribution of the first $3|V|-1$ coordinates of $\Pcal$ is exactly $\Dhost(\Tcal)$.
\end{itemize}
\end{restatable}
We refer to the target distribution as $\Dhost$ because it essentially ``hosts" the following distribution, $\Dhard(n,m)$, defined below:
\begin{definition}[The $\Dhard(n,m)$ Distribution]\label{def:dhard}
A sample $(x,y)\sim\Dhard(n,m)$ is drawn as follows: first sample $x\sim\Ucal_{1/4}^n$ according to the $(1/4)$-biased product distribution. If $x$ has odd Hamming weight, then sample $y\in\bin^m$ uniformly at random; otherwise $x$ has even Hamming weight, and sample $y$ as a uniform $m$-bit string of parity $|x|/2\Mod2$.
\end{definition}

$\Dhard$ is the distribution which we are able to prove classical hardness for in a more straightforward way. 
Observe that the relationship between $\Dhard$ and $\Dhost$ is analogous to that between the PHP and RPHP.
The reduction from $\Dhost$ to $\Dhard$ is given in \Cref{sec:nc0_reduction}.

\begin{restatable}{lemma}{lemmareduction}\label{lem:reduction}
Let $\Tcal=(V,E)$ be a tree and let $v^*\in V$ be arbitrary.
Define $K=\sum_{v\in V}|P_v|$, where $P_v$ is the set of edges on the unique path between $v^*$ and $v$.
Then there exists a $5$-local function $\reduc\colon\bin^{3|V|-1}\times\bin^*\to\bin^{2|V|+K}$ such that
$$
\reduc\pbra{\Dhost(\Tcal),\Ucal_{1/2}^*}=\Dhard(|V|,|V|+K).
$$
\end{restatable}
\subsection{Classical Hardness}\label{sub:overview_hardness}

The classical lower bound of \Cref{thm:main_quantitative} is largely derived from the following hardness result.
Let $\tow(x)$ denote the tower of 2's of height $x$ (e.g., $\tow(3) = 2^{2^2}$).
\begin{restatable}{theorem}{thmnczero}\label{thm:nc0}
Let $d\ge1$ be an integer.
Assume $n\ge\tow(30d)$ and $m\le n^2/\tow(30d)$.
Then any $d$-local distribution has total variation distance at least $0.24$ from $\Dhard(n,m)$.
\end{restatable}

Combining \Cref{thm:nc0} with \Cref{lem:reduction} easily gives the following corollary.

\begin{corollary}\label{cor:nc_iterated}
    Let $\Tcal=(V,E)$ be a tree and let $v^*\in V$ be arbitrary.
    Define $K=\sum_{v\in V}|P_v|$, where $P_v$ is the set of edges on the unique path between $v^*$ and $v$.
    Additionally, let $d\ge1$ be an integer, and assume $|V|\ge\tow(30(d+5))$ and $|V|+K\le |V|^2/\tow(30(d+5))$.
    Then any $d$-local distribution has total variation distance at least $0.24$ from $\Dhost(\Tcal)$.
\end{corollary}
\begin{proof}
    Assume by contradiction there exists a $d$-local function $f\colon \bin^*\to \bin^{3|V|-1}$ and a product distribution $\Pi$ over $\bin^*$ such that the distribution of $f$ applied to samples drawn from $\Pi$, denoted $f(\Pi)$, is $\delta$-close to $\Dhost(\Tcal)$ for some $\delta < 0.24$.
    Define a new function $g\colon \bin^*\to \bin^{2|V|+K}$ by
    \[
        g(\Pi) = \reduc(f(\Pi), \bin^*),
    \]
    where $\reduc$ is defined as in \Cref{lem:reduction}.
    Then $g$ is $(d+5)$-local by \Cref{lem:reduction} and $\delta$-close to $\Dhard(|V|, |V|+K)$ by the data processing inequality.
    This contradicts \Cref{thm:nc0}.
\end{proof}

Before sketching the main ideas behind the proof of \Cref{thm:nc0}, a few remarks are in order. 
First, a tighter analysis can yield distance $\frac14 - \eps$, assuming $m \le O_{\eps}(n^2/\tow(30d))$; this is near optimal, as the $2$-local\footnote{The distribution is $2$-local if the input bits are unbiased coins. When we allow input bits with mixed bias of $1/4$ and $1/2$, the distribution is $1$-local.} distribution $\Ucal_{1/4}^n\times\Ucal_{1/2}^m$ achieves distance $\frac14-o(1)$. 
Second, the quadratic upper bound on $m$ in \Cref{thm:nc0} is necessary; we show $\Dhard(n,m)$ is $O(1)$-local when $m\ge\Omega(n^2)$ in \Cref{prop:nc0_upper}.
Finally, it is necessary that $x\sim \Ucal_{1/4}^n$ and not $x\sim \Ucal_{1/2}^n,$ as the latter \emph{can} be exactly sampled (see \Cref{prop:nc0_upper2}), though any bias other than $0,1,$ or $1/2$ will be hard.

Let us now provide an overview of \Cref{thm:nc0}'s proof.
Fix an arbitrary $d$-local function $f\colon\bin^* \to \bin^{n+m}$ and an arbitrary product distribution $\Pi$ over $\bin^*$ as input. 
Our goal is to show that the distribution $f(\Pi)$ is 0.24-far from $\Dhard(n,m)$.
One immediate challenge in working with $d$-local functions is that the locality constraint is ``one-sided.''
Even though no output bit is influenced by many input bits, there may exist an input bit that affects every single output bit.
The resulting output distribution, then, can have complicated correlations, which muddle the analysis.

\paragraph{The Structured Case: A First Attempt.}
To warm-up, we first consider the idealized setting where there are $r$ ``non-connected'' output bits, by which we mean no two such output bits depend on a common input bit.
In particular, the $r$ marginal distributions of $f(\Pi)$ projected onto the individual coordinates are independent.
Here, one should view $r$ as large, say $\Omega_d(n)$.
We proceed via a concentration vs.\ anticoncentration dichotomy, present in various forms in the works \cite{viola2012complexity, filmus2023sampling, viola2023new, kane2024locality, kane2024locally}.
Specifically, we classify each of the $r$ output bits according to how their corresponding marginal distribution compares to the marginal distribution of the target distribution.

At a high level, we would like to argue that either many of these output bits have marginal distributions which are far from those of the target distribution, in which case we can combine the marginal errors, or many of these output bits are close to the ``correct'' marginal distribution, in which case a more complicated anticoncentration argument shows that a specific potential function highlights a noticeable discrepancy between the two distributions.
To this end, we call an output bit $b$ \textsf{Type-1} if the marginal distribution $f(\Pi)|_{b}$ is $\delta$-far in total variation distance from the marginal distribution $\Dhard(n,m)|_b$, and call it \textsf{Type-2} otherwise.
Here, $\delta = O_d(1)$ is some small threshold parameter. 

Suppose at least $r/2$ of the non-connected output bits are \textsf{Type-1}.
Note that since total variation distance is closed under projection, a single Type-1 neighborhood already gives distance $\delta$.
To strengthen the bound, one can take advantage of independence and apply standard concentration inequalities, as in the proof of \cite[Lemma 4.2]{kane2024locality}, to conclude $f(\Pi)$ has distance roughly $1 - e^{-\delta^2\cdot r}$ from $\Dhard(n,m)$.\footnote{There is a small subtlety here, in that if the set of \textsf{Type-1} output bits fully contains the last $m$ output bits, then those output bits are not a product distribution in $\Dhard(n,m)$. 
For simplicity, we will assume that this does not occur, although the full statement of \cite[Lemma 4.2]{kane2024locality} is robust enough to still apply in that scenario.}
For $r \gg 1/\delta^2$, this is at least 0.24, as desired.

The more involved case is when at least $r/2$ output bits are \textsf{Type-2}.
Here, rather than directly comparing $f(\Pi)$ to $\Dhard(n,m)$, we compare the expectation of a complex-valued potential function $h(x,y)=\complexi^{|x|+2|y|}$ over samples $(x,y)$ drawn from the two distributions.
Direct calculations show that $\E_{x,y}\sbra{h(x,y)} \approx 1/2$ for $(x,y) \sim \Dhard(n,m)$ (see \Cref{clm:dhard_potential}) and that $|\E[\complexi^A]|$ is bounded away from 1 for any integral random variable $A$ suitably far from constant modulo 4 (see \Cref{clm:nc0_single_decay}).
It is tempting to argue that by the independence of the non-connected output bits, we can fix the value of the input bits not affecting any \textsf{Type-2} output bits to view $\E_{x,y}\sbra{h(x,y)}$ as a product of many independent random variables with magnitudes bounded away from 1.
Then we could conclude that for each of these input conditionings, $|\E_{x,y}\sbra{h(x,y)}| \ll 0.01$ for $(x,y) \sim f(\Pi)$, which would give the desired distance of $0.24$ (using \Cref{lem:tvdist_after_conditioning}).

The problem, however, is that the contributions of the remaining output bits can compensate for those of the non-connected output bits.
For example, consider the string $z_1, 1-z_1, z_2, 1 - z_2, \dots, z_k, 1 - z_k$, where the $z_k$'s are independent random bits. 
There are $k$ independent bits, yet the string's Hamming weight is fixed at $k$.
Thus, we cannot reason about $\E_{x,y}\sbra{h(x,y)}$ solely from the non-connected output bits.
Instead, we need to consider the \emph{neighborhood} of each output bit $b$, i.e., the set of output bits that are also influenced by the input bits determining $b$.

\paragraph{The Structured Case: Refining the Output Structure.}
To fix our analysis, let us change our assumption from there being $r$ non-connected output bits to there being $r$ non-connected neighborhoods.
Here, we refer to two neighborhoods $N_1, N_2$ as non-connected if for every pair of output bits $b_1 \in N_1$ and $b_2 \in N_2$, the input bits that determine $b_1$ are disjoint from those that determine $b_2$. 
We can similarly classify each neighborhood as \textsf{Type-1} or \textsf{Type-2} depending on the distance of its marginal distribution to that of the target distribution.
The analysis in the case of many \textsf{Type-1} neighborhoods is performed almost identically to the previous scenario, but now we are able to reason more carefully when there are many \textsf{Type-2} neighborhoods.

Indeed, consider a \textsf{Type-2} neighborhood $N = (x',y')$, where $x'$ is the output bits contained in the first $n$ indices (corresponding to $x$) and $y'$ is the output bits contained in the latter $m$ indices (corresponding to $y$).
If we can show that $|x'|+2|y'|\Mod4$ is not too close to being a constant, then the potential function argument sketched above can actually be carried out.
To this end, let $b$ be the output bit that defines the neighborhood $N = N(b)$, and consider the effect of conditioning on $b=0$ vs.\ on $b=1$.

First suppose $b$ is in the $x$ part.
In this case, we can write $x' = (b, x'')$ and express $|x'|+2|y'|\Mod4$ as $b+|x''|+2|y'|\Mod4$.
Recall that $N$ is a \textsf{Type-2} neighborhood, so it should resemble a product distribution.
In particular, the distribution of $|x''|+2|y'|\Mod4$ conditioned on $b=0$ should be roughly the same as when conditioned on $b=1$.
Observe then, that $1+|x''|+2|y'|\Mod4$ and $|x''|+2|y'|\Mod4$ should have noticeably different distributions, as we are essentially comparing a binomial distribution with its shift.
Since $b$ should be close to $(1/4)$-biased, it takes both values with constant probability, so $|x'|+2|y'|\Mod4$ cannot be too close to any fixed value.
A similar analysis shows that if $b$ is in the $y$ part, then we are comparing the density of $|x'| + 2|y'|\Mod4$ and $|x'| +2(|y'|+1) \Mod4$. 

Unfortunately, there is a problem with this latter case.
Suppose the neighborhood $N$ does not contain any bits in the $x$ part.
Then we are comparing the density of $2|y'|\Mod4$ and $2(|y'|+1) \Mod4$, or equivalently, $|y'| \Mod2$ and $|y'|+1 \Mod2$.
The $y$ part is $(1/2)$-biased, so $|y'| \Mod2$ \emph{can} have the same distribution as $|y'|+1 \Mod2$!
Note that it is this fact which allows for the previously described sampling algorithm for $(X, \textsf{PARITY})$.
Hence, we must make one further refinement to our assumption.

\paragraph{The Structured Case: A Final Adjustment.}
Now instead of simply assuming there are $r$ non-connected neighborhoods, we insist that all $r$ neighborhoods are generated by output bits in the $x$ part.
Moreover, we will only require the non-connectedness property on bits in the $x$ part of the neighborhoods.
This second condition actually makes the analysis more challenging, but we will later see it is necessary for this model case to be obtainable.
We once more redefine \textsf{Type-1} and \textsf{Type-2} neighborhoods; this time we classify neighborhoods based only on their marginals on the first $n$ output bits.
The case of many \textsf{Type-1} neighborhoods essentially works as before (see \Cref{lem:nc0_type-1}), so it remains to address the case where at least $r/2$ of the neighborhoods are \textsf{Type-2}.

To obtain some structure in the $y$ part, we exploit our assumption on the size of $m$.
Since we have $m \le n^2/\tow(30d)$, most pairs of neighborhoods do not intersect in the last $m$ output bits.
Quantitatively, we can find $C \approx r^2 / (md^2) \gg 1$ non-connected \textsf{Type-2} neighborhoods that do not intersect in the $y$ part.
Without loss of generality, assume they are $N(1), N(2), \dots, N(C)$.
By fixing the value of all the input bits that do not affect $1, 2, \dots, C$, the contributions to $h$ from these neighborhoods are now independent.
In particular, the expectation of $h$ becomes a product of expectations over the output of the neighborhood.
As noted above, we can conclude the expectation over the neighborhood $N = (x',y')$ is bounded away from 1 if $|x'|+2|y'|\Mod4$ is not too close to any fixed value.

This ends up being a bit difficult to show directly, since while the $C$ neighborhoods are disjoint in the $y$ part, they may be connected.
Fortunately, the variance of $|x'|+2|y'|\Mod4$ over a random such fixing of the input bits follows from that of $|x'| \Mod2$, where we do have non-connectedness in the $x$ part.
By the previous argument of considering $|x'|$ conditioned on the output bit $b$ being 0 vs.\ being 1, we are able to prove $|x'| \Mod2$ is typically not too close to constant (see \Cref{clm:lem:nc0_type-2_1}). 
This concludes the analysis of many \textsf{Type-2} neighborhoods (see \Cref{lem:nc0_type-2}), as well as the proof of \Cref{thm:main_quantitative} under certain ideal assumptions.

\paragraph{Reduction to the Structured Case.}
Previously, we assumed a rather strong structure: $r = \Omega_d(n)$ many output bits generating neighborhoods that are non-connected in $[n]$.
This, of course, is not a structure readily present in an arbitrary $d$-local function $f$.
To reduce to this case, a standard approach (appearing in, e.g., \cite{viola2012complexity, bravyi2018quantum, viola2020sampling, filmus2023sampling, viola2023new, kane2024locality, kane2024locally}) is to strategically condition on bits to express an arbitrary $d$-local function as a convex combination of functions with the desired structure.
In other words, if we can find some set $S$ of input bits whose removal induces many non-connected neighborhoods of the form we want, then we can express $f(\Pi)$ as
\[
    f(\Pi) = \E_{\rho \in \bin^S} \sbra{f_\rho(\Pi)},
\]
where $f_\rho\colon \bin^* \to \bin^{n+m}$ is defined as $f$ with the input bits in $S$ fixed to their values in $\rho$.
Observe that each $f_\rho$ has the structured form we already know how to analyze, regardless of the actual values the bits in $S$ are set to.
More formally, we have:
\begin{enumerate}
    \item If most of the non-connected neighborhoods are \textsf{Type-1}, then $f_\rho(\Pi)$ is $\approx (1-e^{-\Omega_d(r)})$-far from $\Dhard(n,m)$, and
    \item Otherwise, $\E_{(x,y) \sim \Dhard(n,m)}[h(x,y)] - \E_{(x,y) \sim f_\rho(\Pi)}[h(x,y)] \ge 0.49$.
\end{enumerate}
By a union bound argument (see \Cref{lem:tvdist_after_conditioning}), these results on the conditioned functions can be combined to obtain $\tvdist{f(\Pi) - \Dhard(n,m)} \gtrsim 0.245 - 2^{|S|}\cdot e^{-\Omega_d(r)}$.
Then as long as $r \gg |S|$, we obtain the desired distance bound of 0.24.

At this point, the remaining task is combinatorial.
We construct a bipartite graph whose left side is the first $n$ output bits, and whose right side is the input bits.
Note that each left vertex has maximum degree $d$.
We want to remove $s$ right vertices to obtain $r$ non-connected neighborhoods of the prescribed form, where $r \gg s$.
Ideally, we would like $r$ to be as large as possible to maximize the total variation distance.
Fortunately, this task has already been done for us.
By \cite[Corollary 4.11]{kane2024locality}, we can take $s \ll r$ and $r = \Omega_d(n)$, as desired.

For the sake of completeness, we briefly highlight the main idea behind the proof of \cite[Corollary 4.11]{kane2024locality}.
The key observation is that locality, while only explicitly constraining the left vertices, also constrains the right ones, since it upper bounds the number of edges by $dn$.
Thus while we cannot forbid high-degree right vertices, there cannot be many of them.
This implies that we can ``affordably'' remove all right vertices above a particular degree threshold, and greedily find non-connected vertices on the left side.
A more involved analysis (see \cite[Corollary 4.8]{kane2024locality}) provides better parameters than one could obtain via this naive approach, and an even more involved analysis guarantees non-connected left neighborhoods, rather than just vertices.
Still, the proofs morally operate in a similar way to the strategy described.
This completes the sketch of the proof of \Cref{thm:main_quantitative}; the full details can be found in \Cref{sec:nc0}.

We conclude by remarking that the above analysis is fairly robust, and it allows one to rule out the sampleability of a number of simple distributions by shallow circuits.
Thus, the specific distributions we have chosen to consider are primarily a function of what can be produced by shallow quantum circuits, rather than what can be forbidden for shallow classical ones.

\subsection{Boosting the Separation}\label{sub:overview_boosting}

Combining our results thus far produces a separation with constant total variation distance.
In order to prove the stronger separation in \Cref{thm:main_quantitative}, we consider the distribution $\Dhost(\Tcal)^k = \Dhost(\Tcal)\times \cdots \times\Dhost(\Tcal)$.
Certainly, if our quantum circuit can generate $\Dhost(\Tcal)$, then it can also generate $\Dhost(\Tcal)^k$.
Moreover, we can apply the following direct product theorem implicit in \cite{kane2024locality} (and formalized in \Cref{sec:hardness_amp}) to show the overlap of the target distribution with that produced by classical circuits decays exponentially quickly.

\begin{restatable}[Direct Product Theorem]{theorem}{thmdirectproduct}\label{thm:direct_product}
    Let $d, \ell \ge 1$ be integers, and let $\Dcal$ be a distribution over $\bin^\ell$.
    Suppose that for any $d$-local function $f\colon \bin^* \to \bin^\ell$ and binary product distribution $\Pi$ on $\bin^*$, we have
    \[
        \tvdist{f(\Pi) - \Dcal} \ge \delta.
    \]
    Then for any integer $k \ge 1$, $d$-local function $g\colon \bin^* \to \bin^{\ell k}$, and binary product distribution $\Xi$ on $\bin^*$, we have
    \[
        \tvdist{g(\Xi) - \Dcal^{k}} \ge 1 - 4 \exp\cbra{-\pbra{\frac{\delta^2}{16d\ell}}^{4d\ell} \cdot k}.
    \]
\end{restatable}

The proof of \Cref{thm:direct_product}, much like the proof of \Cref{thm:main_quantitative}, uses a graph elimination result derived in \cite{kane2024locality}.
In this context, such a result allows one to find many independent groups of output bits corresponding to instances of $\Dcal$.
Since the marginal distributions of $\Dcal$ and $f(\Pi)$ disagree on each group, we can use a standard concentration inequality to derive a strong error bound.

\begin{proof}[Proof of \Cref{thm:main_quantitative}]
    Let $n = \tow(40d)$, and let $\Tcal = (V,E)$ be the spanning tree of the $\sqrt{n}\times\sqrt{n}$ square grid obtained by including all the edges in the first row, as well as all the edges in each column.
    Observe that the diameter of $\Tcal$ is $3\sqrt{n}$, so $K$ (as defined in \Cref{cor:nc_iterated}) is at most $3n^{3/2}$.
    In particular, our choice of $n$ guarantees $|V|+K\le |V|^2/\tow(30(d+5))$.
    Thus, \Cref{cor:nc_iterated} implies any $d$-local distribution at least $0.24$-far from $\Dhost(\Tcal)$.
    Applying \Cref{thm:direct_product}, we can boost this error to conclude that any $d$-local distribution has distance from $\Dhost(\Tcal)^k$ at least
    \[
        1 - 4 \exp\cbra{-\pbra{\frac{0.24^2}{16d(3n-1)}}^{4d(3n-1)} \cdot k}.
    \]

    Let $c_d > 0$ be a sufficiently large constant depending only on $d$. 
    For any integer $N \ge c_d$, express $N = k\cdot (3n-1) + r$ with $0\le r < 3n-1$, and define the distribution $\Dcal_N = \Dhost(\Tcal)^k \times 0^r$.
    Since total variation distance is closed under projection, we find that any $d$-local distribution has distance from $\Dcal_N$ at least
    \[
        1 - 4 \exp\cbra{-\pbra{\frac{0.24^2}{16d(3n-1)}}^{4d(3n-1)} \cdot \frac{N-r}{3n-1}} \ge 1 - e^{-N/c_d}
    \]
    for large enough $c_d$.

    We conclude by noting that \Cref{prop:qnc0} gives a depth-$7$ quantum circuit that exactly samples $\Dhost(\Tcal)^k$ on input $\ket{0^{k(5n-1)}}$ by considering the marginal distribution on $k(3n-1)$ specific coordinates.
    By padding with $r$ extra zeros, a similar circuit on $k(5n-1) + r \le 2N$ inputs can also sample $\Dcal_N$.
\end{proof}

\begin{remark}
    Our setting of $\Tcal$ is motivated by common topological choices in implementations.
    One could alternatively set $\Tcal$ to be a balanced binary tree to minimize $K$, but this would not affect the final bound.
\end{remark}

\section{Preliminaries}\label{sec:prelims}

In this section, we collect a number of definitions, notation, and useful results, many of which are taken from \cite{kane2024locality}.

We use $\Cbb, \Rbb, \Zbb, \Nbb$ to denote the set of complex, real, integer, and natural numbers, respectively.
For a positive integer $n$, we use $[n]$ to denote the set $\cbra{1,2,\ldots,n}$, and use $\Zbb/n\Zbb=\cbra{0,1,\ldots,n-1}$ to denote the additive group modulo $n$. 
We use $\complexi$ to denote the imaginary unit satisfying $\complexi^2 = -1$.
For a binary string $x$, we use $|x|$ to denote its Hamming weight.
We use $\log(x)$ to denote the logarithm with base $2$.
For $x\in\Nbb$, we use $\tow(x)$ to denote the power tower of base $2$ and order $x$, where 
$$
\tow(x)=\begin{cases}
1 & x=0,\\
2^{\tow(x-1)} & x\ge1.
\end{cases}
$$

\paragraph*{Asymptotics.}
We use the standard $O(\cdot), \Omega(\cdot), \Theta(\cdot)$ notation, and emphasize that in this paper they only hide universal positive constants that do not depend on any parameter.
Occasionally we will use subscripts to suppress a dependence on particular variable (e.g.,\ $O_d(1)$).

\paragraph*{Probability.}
For $\gamma\in[0,1]$, we use $\Ucal_\gamma$ to denote the $\gamma$-biased distribution, i.e., $\Ucal_\gamma(1)=\gamma=1-\Ucal_\gamma(0)$.
Let $\Pcal$ be a (discrete) distribution. We use $x\sim\Pcal$ to denote a random sample $x$ drawn from the distribution $\Pcal$.
If $\Pcal$ is a distribution over a product space, then we say $\Pcal$ is a product distribution if its coordinates are independent.
In addition, for any non-empty set $S\subseteq[n]$, we use $\Pcal|_S$ to denote the marginal distribution of $\Pcal$ on coordinates in $S$.
For a deterministic function $f$, we use $f(\Pcal)$ to denote the output distribution of $f(x)$ given a random $x\sim\Pcal$.

For every event $\Ecal$, we define $\Pcal(\Ecal)$ to be the probability that $\Ecal$ happens under distribution $\Pcal$, and we use $\Pcal(x)$ to denote the probability mass of $x$ under $\Pcal$.
We will make use the following standard concentration inequality.
\begin{fact}[Chernoff's Inequality]\label{fct:chernoff}
Assume $X_1,\ldots,X_n$ are independent random variables such that $X_i\in[0,1]$ holds for all $i\in[n]$.
Let $\mu=\sum_{i\in[n]}\E[X_i]$.
Then for all $\delta\in[0,1]$, we have
$$
\Pr\sbra{\sum_{i\in[n]}X_i\le(1-\delta)\mu}
\le\exp\cbra{-\frac{\delta^2\mu}2}.
$$
\end{fact}

Let $\Pcal_1,\ldots,\Pcal_t$ be distributions.
Then $\Pcal_1\times\cdots\times\Pcal_t$ is a distribution denoting the product of $\Pcal_1,\ldots,\Pcal_t$.
We also use $\Pcal^t$ to denote $\Pcal_1\times\cdots\times\Pcal_t$ if each $\Pcal_i$ is the same distribution as $\Pcal$.
For a finite set $S$, we use $\Pcal^S$ to emphasize that coordinates of $\Pcal^{|S|}$ are indexed by elements in $S$.
We say a distribution $\Pcal$ is a convex combination, or mixture, of $\Pcal_1,\ldots,\Pcal_t$ if there exist $\alpha_1,\ldots,\alpha_t\in[0,1]$ such that $\sum_{i\in[t]}\alpha_i=1$ and $\Pcal=\sum_{i\in[t]}\alpha_i\cdot\Pcal_i$.
That is, $\Pcal(\Ecal) = \sum_{i\in[t]}\alpha_i\cdot\Pcal_i(\Ecal)$ for every event $\Ecal$.

Let $\Qcal$ be a distribution. We use $\tvdist{\Pcal-\Qcal}=\frac12\sum_x\abs{\Pcal(x)-\Qcal(x)}$ to denote their total variation distance.\footnote{To evaluate total variation distance, we need two distributions to have the same sample space. This will be clear throughout the paper and thus we omit it for simplicity.}
We say $\Pcal$ is $\eps$-close to $\Qcal$ if $\tvdist{\Pcal(x)-\Qcal(x)}\le\eps$, and $\eps$-far otherwise.

\begin{fact}\label{fct:tvdist}
Total variation distance has the following equivalent characterizations:
$$
\tvdist{\Pcal-\Qcal}=\max_{\text{event }\Ecal}\Pcal(\Ecal)-\Qcal(\Ecal)=\min_{\substack{\text{random variable }(X,Y)\\\text{$X$ has marginal $\Pcal$ and $Y$ has marginal $\Qcal$}}}\Pr\sbra{X\neq Y}.
$$
\end{fact}

We will later need the following basic results about total variation distance.
The first says that two distributions on a product space must be far apart if their individual marginals are far apart.
The proof is a straightforward application of Hoeffding’s inequality.

\begin{lemma}[{\cite[Lemma 4.2]{kane2024locality}}]\label{lem:tvdist_after_product}
Let $\Pcal$ and $\Wcal$ be distributions over an $n$-dimensional product space, and let $B\subseteq[n]$ be a non-empty set of size $b$.
Assume
\begin{itemize}
\item $\Pcal|_B$ and $\Wcal|_B$ are product distributions, and
\item $\tvdist{\Pcal|_{\cbra{i}}-\Wcal|_{\cbra i}}\ge\eps$ holds for all $i\in B$. 
\end{itemize}
Then
$$
\tvdist{\Pcal-\Wcal}\ge1-2\cdot e^{-\eps^2b/2}.
$$
\end{lemma}

The second result says that if multiple distributions are either far from a specific distribution in total variation distance or in expectation of a potential function, then so too is any convex combination of those distributions.
It follows from a union bound argument.

\begin{lemma}[{\cite[Lemma 4.7]{kane2024locally}\protect\footnote{This lemma is not present in the most up-to-date arXiv version of \cite{kane2024locally}, but it can be found in \url{https://arxiv.org/pdf/2411.08183v1}, which matches the version originally published in STOC'25.}}]\label{lem:tvdist_after_conditioning}
Let $\Pcal_1,\ldots,\Pcal_\ell$ and $\Qcal$ be distributions. Let $\phi$ be a function with output range $[a,b]$ where $a<b$.
Assume for each $i\in[\ell]$, 
$$
\text{either}\quad
\tvdist{\Pcal_i-\Qcal}\ge1-\eta_1
\quad\text{or}\quad
\E_{X\sim\Qcal}\sbra{\phi(X)}-\E_{X\sim\Pcal_i}\sbra{\phi(X)}\ge\eta_2,
$$
where $\eta_2 \le b-a$.
Then for any distribution $\Pcal$ expressible as a convex combination of $\Pcal_1,\ldots,\Pcal_\ell$, we have
$$
\tvdist{\Pcal-\Qcal}\ge\frac{\eta_2}{b-a}-(\ell+1)\cdot\eta_1.
$$
\end{lemma}
Finally, we will require the following standard fact that two distributions which are close in total variation distance remain close after conditioning.
A proof can be found in \cite[Appendix C]{kane2024locality}.
\begin{fact}\label{fct:mult_apx}
Assume $\Pcal$ is $\eps$-close to $\Qcal$, and let $\Pcal',\Qcal'$ be the distributions of $\Pcal,\Qcal$ conditioned on some event $\Ecal$, respectively. Then for any function $f$,
\[
\tvdist{f(\Pcal')-f(\Qcal')}\le\frac{2\eps}{\Qcal(\Ecal)}.
\]
\end{fact}

\paragraph*{Locality.}
Let $\bin^*$ denote the set of finite length bit strings.
Throughout the paper, we will often be working with functions of the form $g\colon\bin^*\to\bin^n$.
Here, however, we will fix the domain size for concreteness and clarity.
That is, let $f\colon\bin^m\to\bin^n$. For each output bit $i\in[n]$, we use $I_f(i)\subseteq[m]$ to denote the set of input bits that the $i^{\text{th}}$ output bit depends on.
We say $f$ is a $d$-local function if $|I_f(i)|\le d$ holds for all $i\in[n]$.
Define $N_f(i)=\cbra{i'\in[n]\colon I_f(i)\cap I_f(i')\neq\emptyset}$ to be the neighborhood of $i$, which contains all the output bits that have potential correlation with the $i^{\text{th}}$ output bit.

We say output bit $i_1$ is connected to $i_2$ if $I_f(i_1)\cap I_f(i_2)\neq\emptyset$.
We say neighborhood $N_f(i_1)$ is connected to $N_f(i_2)$ if there exist $i_1'\in N_f(i_1)$ and $i_2'\in N_f(i_2)$ such that $I_f(i_1')\cap I_f(i_2')\neq\emptyset$.
As such, every output bit is independent of any non-connected output bit, and the output of a neighborhood has no correlation with any non-connected neighborhood of it.
When $f$ is clear from the context, we will drop subscripts in $I_f(i),N_f(i)$ and simply use $I(i),N(i)$.

In some abuse of standard terminology, we will often discuss the locality of \emph{distributions}.
We may say a certain property holds for $d$-local distributions, by which we mean that property holds for $f(\Pi)$ for every $d$-local function $f\colon \bin^*\to\bin^n$ and binary product distribution $\Pi$ over $\bin^*$.

\paragraph*{Bipartite Graphs.}
We sometimes take an alternative view, using bipartite graphs to model the dependency relations in $f$.
Let $G=(V_1,V_2,E)$ be an undirected bipartite graph.
For each $i\in V_1$, we use $I_G(i)\subseteq V_2$ to denote the set of adjacent vertices in $V_2$.
We say $G$ is $d$-left-bounded if $|I_G(i)|\le d$ holds for all $i\in V_1$.
Define $N_G(i)=\cbra{i'\in V_1\colon I_G(i)\cap I_G(i')\neq\emptyset}$ to be the left neighborhood of $i$.

We say left vertex $i_1$ is connected to $i_2$ if $I_G(i_1)\cap I_G(i_2)\neq\emptyset$.
We say left neighborhood $N_G(i_1)$ is connected to $N_G(i_2)$ if there exist $i_1'\in N_G(i_1)$ and $i_2'\in N_G(i_2)$ such that $I_G(i_1')\cap I_G(i_2')\neq\emptyset$.
When $G$ is clear from the context, we will drop subscripts in $I_G(i),N_G(i)$ and simply use $I(i),N(i)$.

It is easy to see that the dependency relation in $f\colon\bin^m\to\bin^n$ can be visualized as a bipartite graph $G=G_f$ where $[n]$ is the left vertices (representing output bits of $f$) and $[m]$ is the right vertices (representing input bits of $f$), and an edge $(i,j)\in[n]\times[m]$ exists if and only if $j\in I_f(i)$.
The notation and definitions of $I_f(i)$ and $N_f(i)$ are then equivalent to those of $I_G(i)$ and $N_G(i)$.

We will require the following two ``graph elimination'' results of \cite{kane2024locality}.
They first lets us find many non-connected vertices, while the second lets us find many non-connected neighborhoods.

\begin{lemma}[{\cite[Corollary 4.8]{kane2024locality}}]\label{lem:non-adj_vtx}
    Let $\beta,\lambda\ge1$ be parameters (not necessarily constant), and let $G=([n],[m],E)$ be a $d$-left-bounded bipartite graph with $d \ge 1$.
    If 
    $$
    \lambda\ge2d\cdot(2d\beta+1)^{2d},
    $$
    then there exists $S\subseteq[m]$ such that deleting those right vertices (and their incident edges) produces a bipartite graph with $r$ non-connected left vertices satisfying
    $$
    |S|\le\frac r\beta
    \quad\text{and}\quad
    r\ge\frac n\lambda.
    $$
\end{lemma}

\begin{lemma}[{\cite[Corollary 4.11]{kane2024locality}}]\label{lem:non-adj_nbhd}
    Let $\lambda,\kappa\ge1$ be parameters (not necessarily constant),  $F(\cdot)$ be an increasing function, and $G=([n],[m],E)$ be a $d$-left-bounded bipartite graph with $d \ge 1$.
    Define
    \begin{equation*}
    \tilde F(x)=\frac1d\cdot\exp\cbra{32d^4x^2\cdot F(2d\cdot x)}.
    \end{equation*}
    Assume $H(\cdot)$ is an increasing function and $H(x)\ge\tilde F(x)$ for all $x\ge L$ where $L\ge1$ is some parameter not necessarily constant.
    If $H(x)\ge2x$ for all $x\ge L$ and
    \begin{equation*}
    F(x)\ge1
    \text{ holds for all $x\ge1$}
    \quad\text{and}\quad
    \kappa\ge\lambda\ge d\cdot H^{(2d+2)}(L),
    \end{equation*}
    where $H^{(k)}$ is the iterated $H$ of order $k$,\footnote{$H^{(1)}(x)=H(x)$ and $H^{(k)}(x)=H(H^{(k-1)}(x))$ for $k\ge2$.} then there exists $S\subseteq[m]$ such that deleting those right vertices (and their incident edges) produces a bipartite graph with $r$ non-connected left neighborhoods of size at most $t$ satisfying
    $$
        |S|\le\frac r{F(t)}
        \quad\text{and}\quad
        r\ge\frac n\lambda
        \quad\text{and}\quad
        t\le\kappa.
    $$
\end{lemma}

\paragraph*{Classical and Quantum Circuits.}
Throughout this work we will (mostly implicitly) consider Boolean circuits which consist of $\mathsf{AND}$, $\mathsf{OR}$, and $\mathsf{NOT}$ gates. Moreover, we will be primarily concerned with $\nicksclass$ circuits, i.e., those circuits where the number of ingoing wires to any gate in the circuit is bounded by a constant. Further, we shall focus on families of circuits of constant depth. Formally,
\begin{definition}[$\nc$ Circuits]\label{def:nc0}
    Let $\Ccal = \{C_n\}_{n \geq 1}$ be a family of circuits where $C_n$ takes $n$ input bits and produces $m(n)$ output bits for some $m\colon \mathbb{N} \to \mathbb{N}$. $\Ccal$ is said to be an $\nc$ family of circuits if there exists a constant $d$ such that the depth of $C_n$ is at most $d$ for all $n \geq 1$.
\end{definition}
We will occasionally specify a circuit family as \emph{(logspace) uniform}, by which we mean every gate can be specified by a deterministic computation using $O(\log n)$ space.

We will also be interested in the quantum analogue of $\nc$ circuits. Quantum circuits are unitary operators that act on $(\mathbb{C}^{2})^{\otimes n}$ where $\mathbb{C}^2$ is spanned by $\{\ket{0}, \ket{1}\}$ here. An \textit{$n$-qubit quantum state} is any vector $\ket{\psi} \in (\mathbb{C}^{2})^{\otimes n}$ with unit $\ell_2$-norm. For $x \in \{0, 1\}^n$ we use $\ket{x}$ to denote the element $\ket{x_1}\otimes \ket{x_2}\otimes \cdots \otimes \ket{x_n}$, and the set $\{\ket{x}\}_{x \in \{0, 1\}^n}$ will be referred to as the \textit{computational basis}. 

In general, a quantum circuit is any unitary operator obtained by composing several layers of non-overlapping gates from some prescribed set of unitary operators, i.e., the \textit{gate set}. The \textit{depth} of a quantum circuit is the number of layers of gates which make up the circuit. In this work we are interested in the restricted class of quantum circuits called $\qnc$ circuits:
\begin{definition}
    Let $\Ccal = \{C_n\}_{n \geq 1}$ be a family of quantum circuits where $C_n$ acts on $n$ qubits. $\Ccal$ is said to be a $\qnc$ family of circuits if there exist constants $c_1$ and $c_2$ such that for all $n \geq 1$, $C_n$ consists only of gates acting on at most $c_1$ qubits and $C_n$ has depth at most $c_2$.
\end{definition}

While these circuits may in general consist of arbitrary gates of constant locality, we will only consider $\qnc$ circuits with a very particular gate set. These gates are defined as follows:
\begin{itemize}
    \item $\Had$, the Hadamard gate, acts on a single qubit as $\Had\ket{b} = \frac{\ket{0} + (-1)^b\ket{1}}{\sqrt{2}}$ for $b \in \{0, 1\}$
    \item $\CNOT$, the Controlled-Not gate, acts on two qubits as $\CNOT\ket{a}\ket{b} = \ket{a}\ket{a \oplus b}$ for $a, b \in \{0, 1\}$
    \item $\Tof$, the Toffoli gate, acts on three qubits as $\Tof\ket{a}\ket{b}\ket{c} = \ket{a}\ket{b}\ket{(a \land b) \oplus c}$ for $a, b \in \{0, 1\}$
    \item $\CS$, the controlled-phase gate, acts on two qubits as $\CS\ket{a}\ket{b} = \complexi^{a\land b}\ket{a}\ket{b}$ for $a, b \in \{0, 1\}$
\end{itemize}

When physically realizing a quantum circuit the property of \textit{geometric locality} is often very desirable. A quantum circuit is said to be geometrically local if the circuit can be implemented on a 2D grid of qubits with all multi-qubit gates acting on adjacent qubits.

\section[The QNC0 Upper Bound]{The $\qnc$ Upper Bound}\label{sec:qnc0}

In this section, we show how to exactly generate the distribution $\Dhost$ with a shallow quantum circuit.
\propqnczero*

\begin{proof}
Let $v^*\in V$ be arbitrary.
We start with $\ket{0^{|V|}}_X\ket{0^{|V|}}_Z\ket{0^{|E|}}_W\ket{0^{2|V|}}_A$ where $A$ is an ancilla register. The circuit proceeds as follows: first, apply $\Had^{\otimes2|V|}$ on $\ket{0^{2|V|}}_A$, followed by $3$-qubit Toffoli gates to compute $|V|$ $3$-bit \textsf{AND}s on $X$ on uniform inputs.
If  the $X$ and $A$ registers are measured in the computational basis, we obtain $X=x$ and $A=(a,a')$.
Note that $x$ is $(1/4)$-biased and $(a,a')$ is uniform conditioned on $a\land a'=x$ (where the $\land$ is taken bit-wise). Next, we apply $\Had^{\otimes |V|}$ to the $Z$-register to obtain $\ket{+^{|V|}}_Z\ket{0^{|E|}}_{W}$. For each $v \in V$ and each edge $e = (v, u) \in E$ which is incident to $v$ we apply a $\CNOT$ gate from the $Z$-register qubit corresponding to $v$ onto the $W$-register qubit corresponding to edge $e = (u, v)$. Note that each edge qubit is the target of exactly two $\CNOT$ gates. 

Consider a coloring of the edges of the graph such that no two edges which share a vertex are assigned the same color. By Brooks' Theorem any bipartite graph admits such a coloring which uses at most $\Delta$ colors. Since none of the edges of the same color are overlapping, we can apply the corresponding $\CNOT$ gates in two layers. Hence, all $\CNOT$ gates can be applied in depth $2\Delta$. After these $\CNOT$s are applied we are left with
\begin{align*}
    \ket{x} \otimes U_{\Tcal}\bigg{(}\ket{+}^{\otimes |V|}\ket{0^{|E|}}\bigg{)} = \ket{x}\otimes \frac{1}{\sqrt{2^{|V|}}} \sum_{z \in \{0, 1\}^{|V|}}\ket{z}\ket{w(z)},
\end{align*}
where $w(z)_{(u, v)} = z_u \oplus z_v$ for all $(u, v) \in E$ and $U_{\Tcal}$ is the previously described sequence of $\CNOT$s. Observe that $w(z) = w(\overline{z})$ where $\overline{z}_v = z_v \oplus 1$. Hence, this state can be written as
\begin{align*}
    \ket{x}\otimes \frac{1}{\sqrt{2^{|V|}}} \sum_{z \in \{0, 1\}^{|V|}}\ket{z}_{Z}\ket{w(z)} &= \ket{x}\otimes \frac{1}{\sqrt{2^{|V|}}} \sum_{z \in \{0, 1\}^{|V|}, z_1 = 0}(\ket{z} + \ket{\overline{z}})\ket{w(z)}.
\end{align*}
Next, for each qubit of $x$ we apply a controlled-phase gate between it and the corresponding qubit in the $Z$ register. This yields
\begin{align*}
    \ket{x}\otimes \frac{1}{\sqrt{2^{|V|}}} \sum_{z \in \{0, 1\}^{|V|}, z_1 = 0}(\complexi^{\langle x, z \rangle}\ket{z} + \complexi^{\langle x, \overline{z}\rangle}\ket{\overline{z}})\ket{w(z)}.
\end{align*}
Finally, we apply $\Had^{\otimes n}$ to the $Z$ register after which all qubits are measured in the computational basis. A diagram of this circuit is shown in \Cref{fig:final_quantum_circuit}. The Toffoli gates and $U_{\Tcal}$ can be applied in parallel as they act on non-overlapping qubits, yielding a final depth count of $2\Delta + 1$. 

It remains to show that random variables corresponding to the measurement outcomes obtained on the $X, Z, \text{ and } W$ registers are distributed according to $\Dhost$. The state just on the $Z$ register before applying the last layer of Hadamard gates is
\begin{align*}
     \frac{1}{\sqrt{2^{|V|}}} \sum_{z \in \{0, 1\}^{|V|}, z_1 = 0}\complexi^{\langle x, z \rangle }\ket{z} + \complexi^{\langle x, \overline{z} \rangle}\ket{\overline{z}}) &= \frac{1}{\sqrt{2^{|V|}}} \sum_{z \in \{0, 1\}^{|V|}, z_1 = 0}\complexi^{\langle x, z \rangle }(\ket{z} + \complexi^{|x| - 2\langle x, \overline{z}\rangle}\ket{\overline{z}}) \\
     &= \frac{1}{\sqrt{2^{|V|}}} \sum_{z \in \{0, 1\}^{|V|}, z_1 = 0}\complexi^{\langle x, z \rangle }(\ket{z} + (-1)^{|x|/2 -\langle x, \overline{z}\rangle}\ket{\overline{z}}).
\end{align*}
Recall that for any $z \in \{0, 1\}^n$ 
\begin{align*}\Had^{\otimes n} \frac{\ket{z} + (-1)^b\ket{\overline{z}}}{\sqrt{2}} &= \frac{1}{\sqrt{2^{n + 1}}}\sum_{y \in \{0, 1\}^n} ((-1)^{\langle y, z \rangle} + (-1)^{b + \langle y, \overline{z} \rangle})\ket{y}\\
&= \frac{1}{\sqrt{2^{n + 1}}}\sum_{y \in\{0, 1\}^n}(-1)^{\langle y, z \rangle}(1 + (-1)^{b + |y|})\ket{y}.
\end{align*}
The result is a superposition over strings of parity $b$, meaning that measuring after the last layer of Hadamards on the $Z$ register yields a bit string of parity $\frac{|x|}{2} + \langle x, z \rangle$ whenever $|x|$ is even. If $|x|$ is odd then the state on $Z$ before applying the final layer of Hadamards is
\begin{align*}
    \frac{1}{\sqrt{2^{|V|}}}\sum_{z \in \{0, 1\}^{|V|}, z_1 = 0}\complexi^{x, z}(\ket{z} + \complexi^{|x| - 2\langle x, \overline{z}\rangle}\ket{\overline{z}}).
\end{align*}
Note that for any $z \in \{0, 1\}^{|V|}$ and $b \in \{0, 1\}$
\begin{align*}
    \Had^{\otimes n}\frac{\ket{z} + \complexi^{2b + 1}\ket{\overline{z}}}{\sqrt{2}} &= \frac{1}{\sqrt{2^{|V|}}}\sum_{y \in \{0,1\}^{|V|}}\frac{(-1)^{\langle x, z\rangle} + \complexi^{2(b + \langle x, \overline{z}) + 1}}{\sqrt{2}}\ket{y}.
\end{align*}
Hence, when $|x|$ is odd, measuring the $Z$ register will yield a uniformly random outcome - in either case the string, $w$, measured on the $W$-register will satisfy $w_{(u, v)} = z_{u} \oplus z_{v}$ for all $(u, v) \in E$; the final layer of Hadamards on $Z$ will not affect this. This means that the $X, Z, \text{ and } W$ registers of $C\ket{0^{5|V| - 1}}$ are distributed exactly as $\Dhost$ when measured in the computational basis.
\end{proof}
It should be noted that the Toffoli gates in our construction are only used to produce a state whose single-qubit marginals are $(1/4)$-biased, i.e., measure $\ket{0}$ with probability $3/4$ and $\ket{1}$ with probability $1/4$. These Toffoli gates may be replaced by any other gate which produces such a bias, like $R_{Y}(\pi/6) = \begin{pmatrix}
    \sqrt{3/4} & -1/2\\
    1/2 & \sqrt{3/4}
\end{pmatrix}$, and the construction presented would work much the same. In fact, the resulting measurement distribution on all qubits would be exactly $\Dhost$, i.e, we would not need to ignore any qubits. One slightly undesirable property of the $R_{Y}(\pi/6)$ gate (and any other single-qubit gate which generates the desired $1/4$-biased marginal) is that this gate may be used to obtain states whose amplitudes do not have magnitude which squares to a dyadic rational. Explicitly, one can obtain a constant-depth quantum circuit using $R_Y$ and $\Had$ which samples exactly from a product distribution which is only hard for $\nc$ because the one-bit marginal bias is irrational: $\Had R_Y \ket{0} = \frac{\sqrt{3} + 1}{2\sqrt{2}}\ket{0} + \frac{\sqrt{3} - 1}{2\sqrt{2}}\ket{0}$. 

This issue does not arise with the Toffoli gate; the unitary computed by the circuit shown in \Cref{fig:final_quantum_circuit} has entries with magnitudes that square to dyadic rational values. As such, the resulting separation does not rely on any sort of precision limitation inherent in classical sampling circuits. 

Further, our construction can be made geometrically local, where all gates only act on adjacent sets of qubits with the qubits arranged on a 2D grid layout. 
\begin{figure}[ht]
 \centering
 \input{quantum_circuit}
 \caption{A depiction of the $\qnc$ circuit whose measurement distribution on the last $3n - 1$ qubits is exactly $\Dhost(\Tcal)$. Here $U_{\Tcal}$ is the $(2n - 1)$-qubit unitary which acts as $U_{\Tcal}\ket{z}\ket{b} = \ket{z}\bigotimes_{e = (u, v)\in E}\ket{b_e\oplus z_u \oplus z_v}$ - as shown in the proof of \Cref{prop:qnc0} $U_{\Tcal}$ can be implemented via $\CNOT$s in depth $2\Delta$. While the circuit shown above is not geometrically local, a rearrangement of the qubit wires would allow for all $\CNOT$, $\Tof$, and $\CS$ gates to act only on adjacent qubits in a $2D$-grid architecture.}
 \label{fig:final_quantum_circuit}
\end{figure}

\subsection{Toward A Minimal Gate Set?}\label{sec:optimal_gateset}
In the previous construction the gate set used is $\{\Had, \CS, \CNOT, \Tof\}$. Note that $\CNOT$ can be simulated by applying $\Tof$ with an additional ancilla set to $\ket{1}$, i.e.,
\begin{align*}
    \Tof\ket{1}\ket{a, b} = (\mathbb{I}\otimes \CNOT)\ket{1}\ket{a, b}.
\end{align*}
Thus, one can construct a constant-depth circuit consisting only of gates from the set $\{\Had, \CS, \Tof\}$ whose output distribution is not $\nc$-sampleable. This begs the question: \textit{Is this gate set a minimal set for achieving such a separation?} It is well known that $\Had$ and $\Tof$ are sufficient for universal quantum computation \cite{aharonov2003simple}, so it may be tempting to simulate the $\CS$ gate in our construction via $\Had$ and $\Tof$. The standard method for such a simulation involves representing arbitrary $n$-qubit states as $(n + 1)$-qubit states which only have real amplitudes in the following way:
\begin{align*}
    \ket{\psi} = \sum_{x \in \{0, 1\}^n}\alpha_x\ket{x} \to \sum_{x \in \{0, 1\}^n}\ket{x}\otimes(\Re(\alpha_x)\ket{0} + \Im(\alpha_x)\ket{1}) = \ket{\psi'}.
\end{align*}

Indeed, if $C$ is the circuit from \Cref{prop:qnc0} with $C\ket{0^m} = \ket{\psi}$ and $C'$ satisfies $C'\ket{0^{m'}} = \ket{\psi'}$, then the measurement distribution of $\ket{\psi'}$ would be identical to that of $\ket{\psi}$ on the appropriate subset of qubits. Recall that $\CS$ acts as $\CS\ket{x, y} = \complexi^{x \land y}\ket{x, y}$, so the $3$-qubit real unitary which it corresponds to acts as
\begin{align*}
    \CS'\ket{x, y, z} = (-1)^{x\land y \land z}\ket{x, y, z \oplus (x\land y)}.
\end{align*}
In our construction we apply $\CS$ on $n$ disjoint pairs of qubits in a single layer, but doing so with the standard simulation technique of \cite{aharonov2003simple} would require super-constant depth as these $\CS'$ gates would overlap on the last qubit (that qubit which maintains the real and imaginary parts of each amplitude). Hence, it remains unclear if the $\{\Had, \CS, \Tof\}$ gate set is minimal for the separation exhibited here. We leave this direction for future work.

\section[The NC0 Lower Bound]{The $\nc$ Lower Bound}\label{sec:nc0}

The goal of this section is to establish the classical results required for the separation in \Cref{thm:main_quantitative}.
In \Cref{sec:weak_nc_LB}, we prove that any bounded locality distribution must have constant distance from $\Dhard$ (i.e., \Cref{thm:nc0}).
\Cref{sec:nc0_upper} contains some sampling algorithms that justify our parameter choices in \Cref{thm:nc0}.
We then prove \Cref{lem:reduction} in \Cref{sec:nc0_reduction} to reduce the hardness of $\Dhost$ to that of $\Dhard$.
Finally, we formalize a direct product theorem in \Cref{sec:hardness_amp} to boost the distance from constant to $1 - o(1)$.

\subsection{A Weak Lower Bound}\label{sec:weak_nc_LB}

We begin by proving \Cref{thm:nc0}, restated below.

\thmnczero*

Let $f\colon\bin^*\to\bin^{n+m}$ be an arbitrary $d$-local function, and let $\Pi$ be an arbitrary product distribution on $\bin^*$. 
Our goal is to show $f(\Pi)$ is at least $0.24$-far from $\Dhard(n,m)$.

Recall that for each $i\in[n+m]$, $I(i)$ is the set of input bits that the $i^{\text{th}}$ output bit depends on, and $N(i)$ is the neighborhood of the $i^{\text{th}}$ output bit, i.e., the set of output bits sharing common input bits with $i$.
Let $S$ be a subset of input bits. We define $I_S(i)=I(i)\setminus S$ and use $N_S(i)$ to denote the neighborhood of the $i^{\text{th}}$ output bit after fixing the inputs in $S$. Note that these definitions do not depend on how we fix the bits in $S$.

Our first step in proving \Cref{thm:nc0} is to obtain a choice of $S$ such that conditioning on the bits in $S$ reduces $f(\Pi)$ to a more structured distribution.

\begin{lemma}\label{lem:nc0_structure}
There exist
$$
s\le\frac r{2^{20t}},\quad
r\ge\frac n{\tow(20d)},\quad
t\le\tow(20d)
$$
and distinct indices $i_1,\ldots,i_r\in[n]$ and a subset $S$ of size $|S|\le s$ such that the following holds.
\begin{enumerate}
    \item $I_S(u)\cap I_S(v)=\emptyset$ for all distinct $j,j'\in[r]$, any $u\in N_S(i_j)\cap[n]$, and any $v\in N_S(i_{j'})\cap[n]$.
    \item Each $N_S(i_j)\cap[n]$ has size at most $t$.
\end{enumerate}
\end{lemma}
\begin{proof}
Recall we may associate to $f$ a bipartite graph whose right and left parts are the set of input and output bits, respectively.
Here, we take the left part to only consist of the output bits in $[n]$.
The conclusion will then follow from \Cref{lem:non-adj_nbhd}.
Set $F(x) = 2^{20x}$ so that
\[
    \tilde F(x)=\frac1d\cdot\exp\cbra{32d^4x^2\cdot 2^{40dx}}.
\]
Define $H(x) = 2^{2^{2^x}}, L = 10\log(2d)$, and $\kappa = \lambda = \tow(20d)$.
Then
$H(x) \ge \tilde F(x) \ge 2x$ for all $x \ge L$, $F(x) \ge 1$ for all $x \ge 1$, and $\kappa, \lambda \ge H^{(2d+2)(L)}$, so we can apply \Cref{lem:non-adj_nbhd} to conclude the proof.
\end{proof}

For each conditioning $\rho\in\bin^S$ on the bits in $S$, define the restricted function $f_\rho$ as $f$ but with the input bits in $S$ fixed to $\rho$.
We split our analysis into two cases depending on the behavior of the marginal distributions of $N_S(i_j)\cap[n]$ for $i_1, \dots, i_r$ from \Cref{lem:nc0_structure}.
For each $j\in[r]$, we say $i_j$ is \textsf{Type-1} in $f_\rho$ if the marginal distribution of $f_\rho(\Pi)$ on $N_S(i_j)\cap[n]$ is $2^{-5t}$-far from the $(1/4)$-biased product distribution; we say $i_j$ is \textsf{Type-2} in $f_\rho$ otherwise.

We first handle the easy case where \textsf{Type-1} indices are abundant. 

\begin{lemma}\label{lem:nc0_type-1}
If there are at least $r/2$ \textsf{Type-1} indices in $f_\rho$, then $f_\rho(\Pi)$ is $\pbra{1-2 \exp\cbra{-r / 2^{12t}}}$-far from $\Dhard(n,m)$.
\end{lemma}
The proof is similar to that of \cite[Lemma 5.14]{kane2024locality}.
\begin{proof}
By rearranging the indices if necessary, we may assume without loss of generality that $1, 2, \ldots, r/2$ are \textsf{Type-1} indices in $f_\rho$.
That is,
\[
    \tvdist{f_\rho(\Pi)|_{N_S(i) \cap [n]} - \Dhard(n,m)|_{N_S(i) \cap [n]}} \ge 2^{-5t}
\]
for all $i \in [r/2]$.
Let $R = [n] \setminus (N_S(1)\cup \cdots \cup N_S(r/2))$ be the output bits in $[n]$ that are not contained in any of the first $r/2$ neighborhoods.
We will apply \Cref{lem:tvdist_after_product} with $\Pcal, \Wcal, B$ defined as follows:
\begin{itemize}
    \item $\Pcal$ is $f_\rho(\Pi)$ restricted to the output bits in $[n]$, but with each neighborhood and $R$ viewed as individual coordinates. 
    That is,
    \[
        \Pcal = (f_\rho(\Pi)|_{N_S(1)\cap [n]}, f_\rho(\Pi)|_{N_S(2)\cap [n]}, \dots, f_\rho(\Pi)|_{N_S(r/2)\cap [n]}, R)
    \]
    is a distribution over a product space of $(r/2) + 1$ coordinates.
    
    \item $\Wcal$ is the $(1/4)$-biased product distribution over $[n]$, but grouped in the same way as $\Pcal$.
    
    \item $B = [r/2]$.
\end{itemize}
Observe that $\Pcal|_{B}$ and $\Wcal|_{B}$ are both product distributions, since any pair of restricted neighborhoods $N_S(i)\cap [n]$ and $N_S(j)\cap [n]$ for distinct $i,j \in [r/2]$ do not share input bits by \Cref{lem:nc0_structure}.
Thus, we may apply \Cref{lem:tvdist_after_product} with the parameters defined above, as well as the data processing inequality, to conclude
\[
    \tvdist{f_\rho(\Pi) - \Dhard(n,m)} \ge 1-2 \exp\cbra{-(2^{-5t})^2 \cdot r/4} \ge 1-2 \exp\cbra{-r \cdot 2^{-12t}}. \qedhere
\]
\end{proof}

To analyze \textsf{Type-2} indices, we will use the following potential function $h\colon\bin^{n+m}\to\Cbb$:
$$
h(x,y)=\complexi^{|x|}(-1)^{|y|} = \complexi^{|x|+2|y|},\quad\text{where $x\in\bin^n,y\in\bin^m$.}
$$
We will show that $\E[h(x,y)] \approx 1/2$ when $(x,y) \sim \Dhard(n,m)$, but $\E[h(x,y)]$ is far from $1/2$ when $(x,y) \sim f_\rho(\Pi)$.
Later, we will leverage this discrepancy using \Cref{lem:tvdist_after_conditioning}.

\begin{claim}\label{clm:dhard_potential}
Assume $(x,y)\sim\Dhard(n,m)$. Then $\E_{x,y}\sbra{h(x,y)}=\frac12+\pbra{\frac12}^{n+1}$.
\end{claim}
\begin{proof}
If $|x|$ is even, then $h(x,y)\equiv1$; otherwise $\E_y[h(x,y)]=0$.
Thus,
\begin{align*}
    \E_{x,y}\sbra{h(x,y)} = \Pr_{x}\sbra{|x|\text{ is even}} &= \frac{1}{2} + \frac{\Pr_{x}\sbra{|x|\text{ is even}} - \Pr_{x}\sbra{|x|\text{ is odd}}}{2} \\
    &= \frac{1}{2} + \frac{\sum_{i=0}^n \binom{n}{i} (-1)^i \pbra{1/4}^i \pbra{3/4}^{n-i}}{2} \\
    &= \frac{1}{2} + \frac{\pbra{(-1/4) + (3/4)}^n}{2} =  \frac{1}{2} +  \frac{1}{2^{n+1}}. \qedhere
\end{align*}
\end{proof}

To analyze $h$ for $f_\rho(\Pi)$, we will need the following lemma.
It essentially says that two coupled $(1/4)$-biased vectors will differ in Hamming weight modulo 2 a noticeable fraction of time, as long as part of the vectors are independent.
Note that the statement and proof are similar to that of \cite[Lemma 4.4]{kane2024locality}.

\begin{lemma}\label{lem:nc0_resample}
Let $(A,B,C,D)$ be a random variable where $A,C\in\bin$ and $B,D\in\bin^{t-1}$.
Assume
\begin{itemize}
\item $A$ is independent from $(B,D)$ and $B$ is independent from $(A,C)$,
\item $(A,C)$ and $(B,D)$ have the same marginal distribution and are $2^{-5t}$-close to $\Ucal_{1/4}^t$.
\end{itemize}
Then we have
$$
\Pr\sbra{A + |C|\equiv B + |D|\Mod2}\le 1- 2^{-3t}.
$$
\end{lemma}
\begin{proof}
    If $t=1$ then $\Pr\sbra{A+|C|\equiv B+|D|\Mod2}=\Pr\sbra{A=B}$.
    Since $A$ and $B$ are independent and of the same distribution $2^{-5t}$-close to $\Ucal_{1/4}^1$, we have $\Pr[A=1] = \Pr[B=1] \in \sbra{1/4 - 2^{-5t}, 1/4 + 2^{-5t}}$.
    Hence,
    \begin{equation}\label{eq:lem:anticoncentration_after_coupling_2}
    \Pr\sbra{A=B}=\Pr\sbra{A=1}^2+(1-\Pr\sbra{A=1})^2 \le \pbra{\frac{1}{4} - 2^{-5t}}^2 + \pbra{1-\frac{1}{4} + 2^{-5t}}^2 \le 1-2^{-2t}.
    \end{equation}
    
    Now we assume $t\ge2$.
    Expand $\Pr\sbra{A+|C|\equiv B+|D|\Mod 2}$ as
    \begin{align}\label{eq:lem:anticoncentration_after_coupling_1}
    \sum_{a,b\in\bin}\Pr\sbra{A=a,B=b}\Pr\sbra{a+|C|\equiv b+|D|\Mod 2\mid A=a,B=b}.
    \end{align}
    For fixed $a$ and $b$, consider the distribution of $a+|C|\bmod 2$ conditioned on $A=a,B=b$.
    Since $B$ is independent from $(A,C)$, it is the same as the distribution, denoted by $\Pcal_a$, of $a+|C|\bmod 2$ conditioned on $A=a$.
    Similarly define $\Qcal_b$ as the distribution of $b+|D|\bmod 2$ conditioned on $B=b$ (or equivalently, conditioned on $B=b,A=a$).
    
    Since $(A,C)$ is $2^{-5t}$-close to $\Ucal_{1/4}^t$, by \Cref{fct:mult_apx}, $\Pcal_0$ is $(3\cdot 2^{-5t})$-close to $\Dcal_0$, the distribution of $|V|\bmod 2$ for $V\sim\Ucal_{1/4}^{t-1}$.
    Similarly, $\Qcal_1$ is $(8\cdot 2^{-5t})$-close to $\Dcal_1$, the distribution of $1+|V|\bmod 2$ for $V\sim\Ucal_{1/4}^{t-1}$.
    Hence,
    \begin{align*}
        \Pr\sbra{|C|\equiv1+|D|\Mod 2\mid A=0,B=1} &\le1-\tvdist{\Pcal_0-\Qcal_1} \tag{by \Cref{fct:tvdist}} \\
        &\le 1 + 11\cdot 2^{-5t} - \tvdist{\Dcal_0-\Dcal_1}.
    \end{align*}
    Note that 
    \[
        \tvdist{\Dcal_0-\Dcal_1} = \Pr_{V \sim \Ucal_{1/4}^{t-1}}\sbra{|V| \text{ is even}} - \Pr_{V \sim \Ucal_{1/4}^{t-1}}\sbra{|V| \text{ is odd}} = \pbra{\frac{1}{2}}^{t-1},
    \]
    where the final equality follows from an identical calculation to that within the proof of \Cref{clm:dhard_potential}.
    Substituting into the previous inequality yields
    \begin{align*}
        \Pr\sbra{|C|\equiv1+|D|\Mod 2\mid A=0,B=1}  \le 1 + 2^{-5t+4} - 2^{-t+1}.
    \end{align*}
    The same bound holds for $\Pr\sbra{1+|C|\equiv|D|\Mod 2\mid A=1,B=0}$.
    Plugging back into \Cref{eq:lem:anticoncentration_after_coupling_1} and using \Cref{eq:lem:anticoncentration_after_coupling_2}, we have
    \begin{align*}
        \Pr\sbra{A+|C|\equiv B+|D|\Mod 2} &\le \Pr[A=B]+\Pr[A\neq B]\cdot\pbra{1 + 2^{-5t+4} - 2^{-t+1}}  \\
        &\le 1 - 2^{-2t}\cdot\pbra{2^{-t+1} - 2^{-5t+4}}  \\
        &\le 1- 2^{-3t}. \qedhere
    \end{align*}
\end{proof}

We also need the following fact which shows deficiency in taking expectation of $h$ for a slightly biased random source.

\begin{claim}\label{clm:nc0_single_decay}
Let $A$ be an integral random variable. Assume $\max_{a\in\Zbb/4\Zbb}\Pr[A\equiv a\Mod4]\le1-\eta$. Then
$$
\abs{\E\sbra{\complexi^A}}\le1-\frac\eta4.
$$
\end{claim}
\begin{proof}
By a simple averaging argument, we must have $\eta \le 3/4$.
For each $a\in\cbra{0,1,2,3}$, define $p_a=\Pr\sbra{A\equiv a\Mod4}$. 
Assume without loss of generality that $p_1,p_2,p_3\le p_0\le1-\eta$.
Then
\begin{align*}
\abs{\E\sbra{\complexi^A}}^2
&=\abs{(p_0-p_2)+(p_1-p_3)\cdot\complexi}^2
=(p_0-p_2)^2+(p_1-p_3)^2\\
&\le p_0^2+\max\cbra{p_1^2,p_3^2}\le p_0^2+\max\cbra{p_1,p_2,p_3}^2\\
&\le(1-\eta)^2+\eta^2=1-2\eta(1-\eta)\\
&\le1-\frac\eta2,
\end{align*}
where we used $\eta\le3/4$ in the last line.
Thus, $\abs{\E\sbra{\complexi^A}} \le \sqrt{1-\frac\eta2} \le 1-\frac\eta4$.
\end{proof}

Now we show that $\E[h(x,y)]$ cannot be close to $\frac12$ in $f_\rho(\Pi)$ if it has many \textsf{Type-2} indices.

\begin{lemma}\label{lem:nc0_type-2}
If there are at least $r/2$ \textsf{Type-2} indices in $f_\rho$, then we have
$$
\abs{\E_{(x,y)\sim f_\rho(\Pi)}\sbra{h(x,y)}}\le 2\exp\cbra{-\frac{r^2}{m\cdot d^2\cdot 2^{3t+9}}}.
$$
\end{lemma}

\begin{proof}
By rearranging the indices if necessary, we may assume without loss of generality that $1, 2, \ldots, r/2$ are \textsf{Type-2} indices in $f_\rho$.
That is,
\[
    \tvdist{f_\rho(\Pi)|_{N_S(i) \cap [n]} - \Dhard(n,m)|_{N_S(i) \cap [n]}} \le 2^{-5t}
\]
for all $i \in [r/2]$.
We now build a bipartite graph $G$ between $[r/2]$ and $[n+m] \setminus [n]$ by connecting $i\in[r/2]$ and $j\in[n+m] \setminus [n]$ if and only if $I_S(i)\cap I_S(j)\ne\emptyset$, or equivalently, $j\in N_S(i)$.
(Note that this bipartite graph is different from the one we often associate with a local function to visualize its input/output bit dependencies.)

Since $f$ (and thus $f_\rho$) is $d$-local and $I_S(i) \cap I_S(j) = \emptyset$ for all distinct $i,j\in [r/2]$, each $j\in[n+m] \setminus [n]$ has degree at most $d$ in $G$.
Hence $G$ has at most $dm$ edges. Therefore there are at most
\begin{equation}\label{eq:constraint_on_C}
    \frac{4md^2C}{r} \le\frac r4
\end{equation}
indices $i\in[r/2]$ with degree more than $\frac{r}{4dC}$, where $C\ge1$ is a parameter satisfying \Cref{eq:constraint_on_C} to be tuned later.
We discard these high degree indices and continue with the at least $r/4$ remaining ones.
By construction, each remaining index connects to at most $\frac{r}{4dC}$ many $j\in[n+m] \setminus [n]$, and each $j\in[n+m] \setminus [n]$ connects to at most $d$ different $i\in[r/2]$, so we can greedily find $C$ indices $i\in[r/2]$ that have disjoint neighborhoods in $G$.

Without loss of generality, assume these indices are $1,2,\ldots,C$. 
In summary, they satisfy the following conditions.
\begin{enumerate}
\item \label{itm:nbhd_inputs} $I_S(u)\cap I_S(v)=\emptyset$ for all distinct $i,i'\in[C]$, any $u\in N_S(i)\cap[n]$, and any $v\in N_S(i')\cap[n]$.

This comes from \Cref{lem:nc0_structure} directly.
\item \label{itm:nbhd_size} $N_S(i)\cap[n]$ has size at most $t$ for all $i\in[C]$.

This comes from \Cref{lem:nc0_structure} directly.
\item In $f_\rho(\Pi)$, the marginal distribution on $N_S(i)\cap[n]$ is $2^{-5t}$-close to the $(1/4)$-biased product distribution for all $i\in [C]$.

This comes from the definition of a \textsf{Type-2} index.
\item \label{itm:disjoint_neighborhoods} $N_S(i)\cap N_S(i')=\emptyset$ for all distinct $i,i'\in[C]$.

This comes from \Cref{lem:nc0_structure} and the selection procedure above.
\end{enumerate}

For each $i\in[C]$, let $T_i=N_S(i)\cap[n]$ and $T_i'=N_S(i)\cap[n+m] \setminus [n]$; and let $X_i\in\bin^{T_i},X_i'\in\bin^{T_i'}$ be the output bits of $f_\rho(\Pi)$ in $T_i,T_i'$ respectively.
Additionally, define $R=[n]\setminus\pbra{\bigcup_{i\in[C]}T_i}$ and $R'=([n+m] \setminus [n])\setminus\pbra{\bigcup_{i\in[C]}T_i'}$; and let $Y\in\bin^{R},Y'\in\bin^{R'}$ be the output bits of $f_\rho(\Pi)$ in $R,R'$ respectively.
Then for $(x,y)\sim f_\rho(\Pi)$, we have
$$
h(x,y)=\pbra{\prod_{i\in[C]}\complexi^{|X_i|+2|X_i'|}}\cdot\complexi^{|Y|+2|Y'|}.
$$
Let $J$ be the set of input bits outside $\bigcup_{i\in[C]}I(i)$, so that fixing the bits in $J$ will fix $Y$ and $Y'$.
For each $\sigma\in\bin^J$ and $i\in[C]$, define
$$
p_{\sigma,i}=\max_{a\in\Zbb/4\Zbb}\Pr\sbra{|X_i|+2|X_i'|\equiv a\Mod4\mid\sigma}.
$$
Note that if $\sigma$ is fixed, then the $(X_i,X_i')$'s are pairwise independent by \Cref{itm:disjoint_neighborhoods}.
Hence we have
\begin{align}
\abs{\E_{(x,y)\sim f_\rho(\Pi)}\sbra{h(x,y)\mid\sigma}}
&=\abs{\E_{X_i,X_i',\forall i\in[C]}\sbra{\prod_{i\in[C]}\complexi^{|X_i|+2|X_i'|}\mid\sigma}}
\tag{since $\sigma$ fixes $Y,Y'$}\\
&=\prod_{i\in[C]}\abs{\E_{X_i,X_i'}\sbra{\complexi^{|X_i|+2|X_i'|}\mid\sigma}}
\tag{by independence}\\
&\le\prod_{i\in[C]}\pbra{1-\frac{1-p_{\sigma,i}}4}
\tag{by \Cref{clm:nc0_single_decay}}\\
&\le\exp\cbra{-\frac14 \sum_{i\in[C]}\pbra{1-p_{\sigma,i}}}, \label{eq:E_of_h_UB}
\end{align}
where the final inequality uses $1-c\le e^{-c}$.

It remains to show that $\sum_{i\in[C]}\pbra{1-p_{\sigma,i}}$ is typically not too small (i.e., for most restrictions $\sigma$ and indices $i$, the value of $|X_i|+2|X_i'| \Mod{4}$ has reasonable variance).
For this, we first analyze the related quantities
\[
    q_{\sigma,i} = \max_{a\in\Zbb/2\Zbb}\Pr\sbra{|X_i|\equiv a\Mod2\mid\sigma}.
\]
The benefit of working with $q_{\sigma,i}$'s is that they are independent with respect to randomly chosen $\sigma$, since the $X_i$'s depend on disjoint sets of input bits by \Cref{itm:nbhd_inputs}.
This will be necessary later to apply standard concentration inequalities.
Note that such independence may not hold for $p_{\sigma,i}$'s, as we have only shown that the neighborhoods $N_S(i)$ are pairwise disjoint for $i\in [C]$, but they could still depend on common input bits.
The upshot is that we can easily relate $p_{\sigma,i}$ and $q_{\sigma,i}$.
Fix an arbitrary $\sigma \in \bin^J$ and index $i\in [C]$, and let $a \in \Zbb/4\Zbb$ maximize $p_{\sigma,i}$.
Observe that we can write $a = b_1 + 2b_2$ for some $b_1,b_2 \in \bin$.
Then,
\begin{align}
    p_{\sigma,i} &= \Pr\sbra{|X_i|+2|X_i'|\equiv a\Mod4\mid\sigma} \notag \\
    &= \Pr\sbra{|X_i|\equiv b_1 + 2(b_2 - |X_i'|)\Mod4\mid\sigma} \notag \\
    &\le \Pr\sbra{|X_i|\equiv b_1\Mod2\mid\sigma} \le q_{\sigma,i}. \label{eq:p_bound_via_q_bound} 
\end{align}

We are now free to focus on analyzing $q_{\sigma,i}$.
From here, the remainder of the proof is similar to that of \cite[Lemma 5.15]{kane2024locality}.
Below, we will consider $\sigma$ as being sampled according to the marginal distribution of $\Pi$ projected onto the coordinates in $J$.

\begin{claim}\label{clm:lem:nc0_type-2_1}
For each $i\in[C]$, we have $\E_\sigma\sbra{\pbra{q_{\sigma,i}}^2}\le1-2^{-3t}$.
\end{claim}
For clarity, we will complete the proof of \Cref{lem:nc0_type-2} before proving \Cref{clm:lem:nc0_type-2_1}.
By Jensen's inequality, \Cref{clm:lem:nc0_type-2_1} implies
\[
    \E_\sigma\sbra{1-q_{\sigma, i}} \ge 1 - \sqrt{\E_\sigma\sbra{\pbra{q_{\sigma,i}}^2}} \ge 2^{-3t-1}.
\]
As noted above, the $q_{\sigma, i}$'s are pairwise independent, so we may apply Chernoff's inequality (\Cref{fct:chernoff}) with $\delta = 1/2$ to obtain
\begin{equation}\label{eq:sigma_non_constant}
    \Pr_\sigma \sbra{\sum_{i\in[C]}\pbra{1-q_{\sigma,i}} \le \frac{1}{2} \cdot 2^{-3t-1}\cdot C} \le \exp\cbra{-\frac{C}{2^{3t+4}}}.
\end{equation}
We say $\sigma$ is \emph{bad} if the above event occurs and \emph{good} otherwise.
Then,
\begin{align*}
    \abs{\E_{(x,y)\sim f_\rho(\Pi)}\sbra{h(x,y)}} &\le \E_\sigma \abs{\E_{(x,y)\sim f_\rho(\Pi)}\sbra{h(x,y)\mid\sigma}} \tag{by triangle inequality} \\
    &\le \Pr_\sigma\sbra{\sigma \text{ is bad}} + \E_{\text{good }\sigma} \abs{\E_{(x,y)\sim f_\rho(\Pi)}\sbra{h(x,y)\mid\sigma}} \\
    &\le \exp\cbra{-\frac{C}{2^{3t+4}}} + \E_{\text{good }\sigma}\sbra{\exp\cbra{-\frac14 \sum_{i\in[C]}\pbra{1-p_{\sigma,i}}}} \tag{by \Cref{eq:sigma_non_constant} and \Cref{eq:E_of_h_UB}} \\
    &\le \exp\cbra{-\frac{C}{2^{3t+4}}} + \E_{\text{good }\sigma}\sbra{\exp\cbra{-\frac14 \sum_{i\in[C]}\pbra{1-q_{\sigma,i}}}} \tag{by \Cref{eq:p_bound_via_q_bound}} \\
    &\le \exp\cbra{-\frac{C}{2^{3t+4}}} + \exp\cbra{-\frac{1}{4}\cdot \frac{1}{2} \cdot 2^{-3t-1}\cdot C} \tag{by def'n of good $\sigma$} \\
    &= 2\exp\cbra{-\frac{C}{2^{3t+4}}}.
\end{align*}
To complete the proof it remains to set the value of $C$.
We would like to choose $C$ as large as possible subject to the constraint in \Cref{eq:constraint_on_C}; that is
\[
    C = \floorbra{\frac1m \cdot \pbra{\frac{r}{4d}}^2} \ge \frac1{2m} \cdot \pbra{\frac{r}{4d}}^2,
\]
which is at least 1 by our assumption on $m$ and \Cref{lem:nc0_structure}.
Plugging into the previous inequality, we conclude
\[
    \abs{\E_{(x,y)\sim f_\rho(\Pi)}\sbra{h(x,y)}} \le  2\exp\cbra{-\frac{C}{2^{3t+4}}} \le 2\exp\cbra{-\frac{r^2}{m\cdot d^2\cdot 2^{3t+9}}}. \qedhere
\]
\end{proof}

We now complete the proof of \Cref{clm:lem:nc0_type-2_1}, showing that $|X_i|$ is unlikely to be fixed modulo 2 by a random $\sigma$.
The proof is similar to \cite[Claim 5.16]{kane2024locality}.

\begin{proof}[Proof of \Cref{clm:lem:nc0_type-2_1}]
By \Cref{itm:disjoint_neighborhoods}, the distribution of $|X_i|$ conditioned on $\sigma$ is the same as the distribution conditioned on all input bits outside of $I(i)$, denoted $\overline{I(i)}$.
Let $Z = (Z_1, Z_2, \dots)$ denote the input bits to $f_\rho$.
Then we can write
\begin{equation}\label{eq:lem:nc0_type-2_1_1}
    \E_\sigma\sbra{\pbra{q_{\sigma,i}}^2} = \E_{Z_j : j\in \overline{I(i)}}\sbra{\max_{a\in\Zbb/2\Zbb}\Pr\sbra{|X_i|\equiv a\Mod2 \mid Z_j : j\in \overline{I(i)}}^2}.
\end{equation}
Consider a new input $Z'$ obtained from $Z$ by resampling the bits in $I(i)$, and let $\widetilde{X}_i$ be the new output neighborhood on the same indices as $X_i$.
Then for any $a\in \Zbb / 2\Zbb$, we have
\begin{align*}
    \Pr\sbra{|X_i|\equiv a\Mod2 \mid Z_j : j\in \overline{I(i)}}^2 &= \Pr\sbra{|X_i|\equiv  |\widetilde{X}_i| \equiv a\Mod2\mid Z_j : j\in \overline{I(i)}} \tag{$X_i$ and $\widetilde{X}_i$ are conditionally independent} \\
    &\le \Pr\sbra{|X_i|\equiv  |\widetilde{X}_i|\Mod2\mid Z_j : j\in \overline{I(i)}}.
\end{align*}
Substituting into \Cref{eq:lem:nc0_type-2_1_1} gives $\E_\sigma\sbra{\pbra{q_{\sigma,i}}^2} \le \Pr\sbra{|X_i|\equiv |\widetilde{X}_i|\Mod2}$.

We conclude by applying \Cref{lem:nc0_resample} with $A = (X_i)|_i, B = (\widetilde{X}_i)|_i, C = (X_i)|_{(N_S(i) \cap [n]) \setminus \{i\}},$ and $D = (\widetilde{X}_i)|_{(N_S(i) \cap [n]) \setminus \{i\}}$.
Note that the conditions of the lemma are met, since resampling the input bits in $I(i)$ decouples $A$ from $(B,D)$ and $B$ from $(A,C)$, and $(A,C)$ and $(B,D)$ have the same marginal distribution $2^{-5t}$-close to the $(1/4)$-biased product distribution.
Thus,
\[
    \E_\sigma\sbra{\pbra{q_{\sigma,i}}^2} \le \Pr\sbra{|X_i|\equiv |\widetilde{X}_i|\Mod2} \le 1-2^{-3t},
\]
where we used \Cref{itm:nbhd_size} to bound the size of $N_S(i) \cap [n]$.
\end{proof}

At this point, we are ready to prove \Cref{thm:nc0}.

\begin{proof}[Proof of \Cref{thm:nc0}]
Recall our goal is to show that the distribution $\Dhard(n,m)$ is $0.24$-far from $f(\Pi) = \E_\rho \sbra{f_\rho(\Pi)}$.
This will follow from \Cref{lem:tvdist_after_conditioning} by showing that each restricted function $f_\rho(\Pi)$ is either far from $\Dhard(n,m)$ in total variation distance or in expectation of the potential function $h(x,y)$.
Fix an arbitrary $\rho$, and consider the $r$ indices guaranteed by \Cref{lem:nc0_structure}.
If at least $r/2$ of them are $\textsf{Type-1}$, then
\[
    \tvdist{f_\rho(\Pi) - \Dhard(n,m)} \ge 1-2 \exp\cbra{-\frac{r}{2^{12t}}}
\]
by \Cref{lem:nc0_type-1}.
Otherwise, at least $r/2$ of them are $\textsf{Type-2}$, and we find that 
\begin{align*}
    \E_{(x,y)\sim\Dhard(n,m)}\sbra{h(x,y)}-\E_{(x,y)\sim f_\rho(\Pi)}\sbra{h(x,y)} &\ge \frac{1}{2} + \pbra{\frac12}^{n+1} - 2\exp\cbra{-\frac{r^2}{m\cdot d^2\cdot 2^{3t+9}}} \\
    &\ge \frac{1}{2} - 2\exp\cbra{-\frac{r^2}{m\cdot d^2\cdot 2^{12t}}}
\end{align*}
by \Cref{clm:dhard_potential} and \Cref{lem:nc0_type-2}.
Thus \Cref{lem:tvdist_after_conditioning} yields
\begin{align*}
     \tvdist{f(\Pi) - \Dhard(n,m)} &\ge \frac{1}{2}\pbra{\frac{1}{2} - 2\exp\cbra{-\frac{r^2}{m\cdot d^2\cdot 2^{12t}}}}  - \pbra{2^{|S|}+1}\cdot 2 \exp\cbra{-\frac{r}{2^{12t}}} \\
     &\ge \frac{1}{4} - \exp\cbra{-\frac{r^2}{m\cdot d^2\cdot 2^{12t}}} - \exp\cbra{3|S|-\frac{r}{2^{12t}}}.
\end{align*}
Recall from our assumption on $m$ and the bounds on $r, |S|, t$ given by \Cref{lem:nc0_structure} that
$$
m \le \frac{n^2}{\tow(30d)}, \quad
|S|\le\frac r{2^{20t}},\quad
r\ge\frac n{\tow(20d)},\quad
t\le\tow(20d).
$$
Continuing the previous chain of inequalities, we have
\begin{align*}
    \tvdist{f(\Pi) - \Dhard(n,m)} &\ge \frac{1}{4} - \exp\cbra{-\frac{(n/\tow(20d))^2}{m\cdot d^2\cdot 2^{12\cdot\tow(20d)}}} - \exp\cbra{3\pbra{\frac r{2^{20t}}}-\frac{r}{2^{12t}}} \\
    &\ge \frac{1}{4} - \exp\cbra{-\frac{1}{m}\cdot \frac{n^2}{\tow(25d)}} - \exp\cbra{-\frac{r}{2^{15t}}} \\
    &\ge \frac{1}{4} - \exp\cbra{-\frac{\tow(30d)}{\tow(25d)}} - \exp\cbra{-\frac{n}{\tow(25d)}},
\end{align*}
which is at least $0.24$ for $n \ge \tow(30d)$.
\end{proof}

\subsection[The NC0 Upper Bounds]{The $\nc$ Upper Bounds}\label{sec:nc0_upper}

In this subsection, we provide sampling schemes that highlight the necessity of the constraints in \Cref{thm:nc0}.
We begin by showing that we must take $m \le O(n^2)$; otherwise, one may produce the distribution $\Dhard(n,m)$ with constant locality. 

\begin{proposition}\label{prop:nc0_upper}
If $m\ge\binom{n+1}2$, then $\Dhard(n,m)$ is a $6$-local distribution with unbiased random bits as inputs.
\end{proposition}
\begin{proof}
We describe the sampling algorithm as follows.
\begin{itemize}
\item Sample a uniform $m$-bit string $z$ of even Hamming weight. This is $2$-local by sampling as in the $\textsf{PARITY}$ example in \Cref{sub:q_sampling} (e.g., output $r_1 \oplus r_2, r_2 \oplus r_3, \ldots, r_{m-1}\oplus r_m, r_m \oplus r_1$ where $r_1, r_2, \ldots$ are unbiased random bits from the input).
\item Sample $x\sim\Ucal_{1/4}^n$ and define $\tilde x\in\bin^{n+1}$ by setting $\tilde x=(x,b)$ where $b$ is an unbiased random bit. This is 2-local.
\item Prepare an $m$-bit string $w$ by putting the products $\tilde x_i\tilde x_j$ for all pairs $1\le i<j\le n+1$ on the first $\binom {n+1}2$ entries (in any order) and padding the rest with zeros. Define $y=z\oplus w$ and output $(x,y)$.
\end{itemize}

Since each bit of $y$ depends on one bit of $z$ and at most two bits of $\tilde x$, the total locality of the above construction is $6$.
Let us now verify correctness. 
It is clear that $x$ has the correct distribution, so it remains to check $y$.
The parity of $|y|$ is given by
\[
    |z| + |w| \equiv \sum_{1\le i<j\le n+1} \tilde x_i\tilde x_j \equiv \sum_{1\le i<j\le n}  x_i x_j + \sum_{i=1}^n x_i b \equiv \binom{|x|}{2} + b|x| \pmod{2}.
\]
Observe that for any fixed $x\in\bin^n$ of odd Hamming weight, flipping the unbiased random bit $b$ flips the parity of $|y|$, so the parity of $|y|$ must also be unbiased.
Furthermore, for any fixed $x\in\bin^n$ of even Hamming weight, the parity of $|y|$ is equal to the parity of $\binom{|x|}{2} \equiv |x|/2 \Mod{2}$.
We conclude by noting that the addition of $z$ acts to ``symmetrize'' the distribution over $y$; each output string is equally likely, conditioned on the Hamming weight having the correct parity.
\end{proof}

Note that if we are willing to add several of the $\tilde x_i\tilde x_j$ terms to each $y$ entry, we can sample $\Dhard(n,m)$ with a $d$-local function of unbiased random bits so long as $m$ is at least a sufficiently large constant multiple of $n^2/d$.

We now show that in the definition of $\Dhard(n,m)$ it is necessary for $x \sim \Ucal_{1/4}^n$ rather than $\Ucal_{1/2}^n$, which a priori may appear a more natural choice.
Below, let $\Dhard^*(n,m)$ be the analogous version of $\Dhard(n,m)$ where $x\sim \Ucal_{1/2}^n$.
That is, a sample $(x,y) \sim \Dhard^*(n,m)$ is drawn by first sampling $x \sim \Ucal_{1/2}^n$ and then choosing $y$ to be a uniform random $m$-bit string when $|x|$ is odd and otherwise a uniform random $m$-bit string with parity $|x|/2\Mod{2}$.

\begin{proposition}\label{prop:nc0_upper2}
If $m\ge n-1$, then $\Dhard^*(n,m)$ is a $6$-local distribution with unbiased random bits as inputs.
\end{proposition}
\begin{proof}
We describe the sampling algorithm as follows.
\begin{itemize}
\item Sample a uniform $n$-bit string $x_\text{odd}$ of odd Hamming weight, and a uniform random $m$-bit string $y_\text{odd}$.
The former is 2-local, and the latter is 1-local. 
\item Sample a uniform $n$-bit string $x_\text{even}$ of even Hamming weight by setting $(x_\text{even})_i = r_i \oplus r_{i+1}$ where $r_1 = r_{n+1} = 0$ and $r_2, r_3, \ldots, r_n$ are unbiased random bits from the input. This is 2-local.
\item Sample a uniform $m$-bit string $z$ of even Hamming weight. This is $2$-local.
\item Prepare an $m$-bit string $w$ by putting $r_i \oplus r_{i}r_{i+1}$ for all $2 \le i \le n$ on the first $n-1$ entries (in any order) and padding the rest with zeros.
Define $y_\text{even}=z\oplus w$. 
\item Sample a uniform random bit $b$ and output $(x_\text{even},y_\text{even})$ if $b=0$ and $(x_\text{odd},y_\text{odd})$ otherwise.
\end{itemize}

Each output bit depends on a bit of $x_\text{odd}$ or $y_\text{odd}$, a bit of $x_\text{even}$ or $y_\text{even}$, and $b$, and so depends on at most $6$ input bits.
Via a similar analysis to that of \Cref{prop:nc0_upper}, the correctness of this sampling scheme will follow from proving the distribution of $|y|$'s parity is correct.
The case of odd $|x|$ is immediate, so let us consider the case of a fixed $x \in \bin^n$ with even Hamming weight.
Here, the parity of $|y|$ is given by
\begin{align*}
    |z| + |w| &\equiv \sum_{i=2}^{n} r_i \oplus r_{i}r_{i+1} \pmod{2} \\
    &\equiv \sum_{i \in [n]} r_i + \sum_{i \in [n]} r_{i}r_{i+1}  \pmod{2} \tag{since $r_1 = 0$} \\
    &\equiv \frac{\sum_{i\in[n]}(r_i+r_{i+1}-2r_ir_{i+1})}2\pmod2 \tag{since $r_1 = r_{n+1}$} \\
    &\equiv \frac{\sum_{i\in[n]}(r_i\oplus r_{i+1})}2\pmod 2 \\
    &\equiv \frac{|x|}2 \pmod 2. \tag*{\qedhere}
\end{align*}
\end{proof}

Note that $\Dhard^*(n,m)$ can in fact be shown to be $5$-local by reusing the bits used to compute $y_\text{even}$ and $y_\text{odd}$.

\begin{remark}\label{rmk:improved_upper}
    The sampling schemes in both \Cref{prop:nc0_upper} and \Cref{prop:nc0_upper2} can be extended to smaller $m$, even beyond the improvement mentioned after \Cref{prop:nc0_upper}, at the cost of increased locality.
    For example, suppose $m \ge n - C$ for some arbitrary integer $C \ge 1$.
    Then we may sample $\Dhard^*(n,m)$ as follows:
    \begin{itemize}
        \item Sample $x_1 \sim \Ucal_{1/2}^{C-1}$.
        \item If $|x_1|$ is even, then sample $(x_2,y)$ where $x_2 \sim \Ucal_{1/2}^{n-C+1}$ and $y$ is a uniform random $m$-bit string when $|x_2|$ is odd and otherwise a uniform random $m$-bit string with parity $(|x_1| + |x_2|)/2\Mod{2}$.
        \item If $|x_1|$ is odd, then sample $(x_2,y)$ where $x_2 \sim \Ucal_{1/2}^{n-C+1}$ and $y$ is a uniform random $m$-bit string when $|x_2|$ is even and otherwise a uniform random $m$-bit string with parity $(|x_1| + |x_2|)/2\Mod{2}$.
        \item In either case, set $x = x_1 \circ x_2$, where $\circ$ denotes concatenation, and output $(x,y)$.
    \end{itemize}
    This requires $C-1$ bits of locality to determine whether to perform the second or third item.
    Observe that $m \ge (n-C+1)-1$, so \Cref{prop:nc0_upper2} provides a constant locality sampling procedure for either item.
    Thus, $\Dhard^*(n,m)$ can be sampled with $C + O(1)$ bits of locality.

    Similarly, $\Dhard(n,m)$ with $m \ge \binom{n+1}{2} - C$ can be sampled with $C + O(1)$ bits of locality using \Cref{prop:nc0_upper}.
    In particular, $\Dhard(n,m)$ (and $\Dhard^*(n,m)$) can be sampled by $\ac$ circuits for any choice of $n,m$.
    This is in stark contrast to the setting of relational problems, where the Parity Halving Problem (i.e., the inspiration for $\Dhard(n,m)$) cannot be computed by functions in $\ac$ \cite{watts2019exponential}.
\end{remark}

\subsection[The NC0 Reduction]{The $\nc$ Reduction}\label{sec:nc0_reduction}

In this subsection, we provide a constant locality reduction from the $\Dhost$ distribution (recall \Cref{def:dhost}) to the $\Dhard$ distribution.
In conjunction with \Cref{thm:nc0}, this implies that $\Dhost$ is also difficult to classically sample.

\lemmareduction*
\begin{proof}
Let $(X,Y,W)$ be a sample from $\Dhost(\Tcal)$, which also implicitly samples $Z$ in \Cref{def:dhost}.
Observe that
\begin{align*}
\abra{Z,X}\Mod2
&=\bigoplus_{v\in V}Z_vX_v
=\bigoplus_{v\in V}X_v\pbra{Z_{v^*}\oplus\bigoplus_{e\in P_v}W_e}
\tag{by def'n of $W$ and $P_v$}\\
&=\pbra{\bigoplus_{v\in V}X_vZ_{v^*}}\oplus\pbra{\bigoplus_{v\in V,e\in P_v}X_vW_e}.
\end{align*}

Sample an unbiased coin $b$.
We will need the following random strings:
\begin{itemize}
\item $Y'$ is a uniform $|V|$-bit string of parity $b$. This is $3$-local with bounded fan-out.
\item $\tilde Y$ is a uniform $K$-bit string of parity $b\oplus\bigoplus_{v\in V,e\in P_v}X_vW_e$. Given $X$, $W$, and $b$, this is $5$-local with bounded fan-out.
\end{itemize}

Now we show that $(X,Y\oplus Y',\tilde Y)$ is distributed as $\Dhard(|V|,|V|+K)$:
\begin{itemize}
\item If $X$ has even Hamming weight, then $Y$ has parity $\abra{Z,X}+|X|/2\Mod2$ as $(X,Y,W)\sim\Dhost(\Tcal)$. Note that in this case $\bigoplus_{v\in V}X_vZ_{v^*}\equiv0$ and thus $\abra{Z,X}\equiv b \oplus |\tilde Y|\Mod2$. Hence $(Y\oplus Y',\tilde Y)$ has parity $|X|/2\Mod2$.

Since $b$ re-randomizes\footnote{This is necessary as $Y$ is always even if $X\equiv0^V$.} the parity of $Y$, $(Y\oplus Y',\tilde Y)$ is uniform with parity $|X|/2\Mod2$.
\item If $X$ has odd Hamming weight, then $Y$ is simply uniform as $(X,Y,W)\sim\Dhost(\Tcal)$ and hence $Y\oplus Y'$ is uniform.
In addition, since $b$ re-randomizes\footnote{This is also necessary as $\bigoplus_{v\in V,e\in P_v}X_vW_e$ may be forced to zero for some $X$.} the parity of $\tilde Y$, $\tilde Y$ is independent from $Y$ and is also uniform.
\end{itemize}
This concludes the proof.
\end{proof}

\subsection[Hardness Amplification for Sampling in NC0]{Hardness Amplification for Sampling in $\nc$}\label{sec:hardness_amp}

In this subsection, we formalize a direct product theorem for sampling in $\nc$ (\Cref{thm:direct_product}), which is largely implicit in \cite{kane2024locality}.
This allows us to ``amplify'' the hardness from \Cref{cor:nc_iterated} by taking multiple copies.
Note that a similar theorem for read-once branching programs is also known \cite{chattopadhyay2022space}.
For a distribution $\Dcal$, recall $\Dcal^{k} = \Dcal \times \cdots \times \Dcal$ denotes the $k$-fold product distribution of $\Dcal$.
We restate \Cref{thm:direct_product} below for convenience.

\thmdirectproduct*

\begin{proof}
    Fix a function $g\colon \bin^* \to \bin^{\ell k}$ and a binary product distribution $\Xi$.
    By viewing the $\ell k$ output bits as $k$ ``chunks'' of $\ell$ consecutive bits, $g$ becomes a $(d\ell)$-local function with $k$ output symbols.
    Let $g_i\colon \bin^* \to \bin^{\ell}$ denote the $i^{\text{th}}$ such symbol, and recall that $g_i(\Xi)$ is $\delta$-far from $\Dcal$ by assumption.
    To obtain a stronger bound for $g(\Xi)$, we will reduce to the case where many $g_i$'s are independent.
    
    By recalling the graph theoretic view of local functions, we may apply \Cref{lem:non-adj_vtx} with $\beta = 4/\delta^2$ and $\lambda = (4d\ell\beta)^{2d\ell+1}$. 
    This guarantees a set $S$ such that any fixing $\rho$ of the input bits in $S$ reduces $g$ to a $d\ell$-local function $g_\rho$ with $r$ non-connected output symbols, where 
    \[
        |S| \le \frac{\delta^2 \cdot r}{4} \quad\text{and}\quad r \ge \frac{k}{(16d\ell/\delta^2)^{2d\ell + 1}}.
    \]
    We then apply \Cref{lem:tvdist_after_product} for each $\rho$ to deduce
    \[
        \tvdist{g_\rho(\Xi) - \Dcal^{k}} \ge 1 - 2\exp\cbra{-\frac{\delta^2 \cdot r}{2}}.
    \]
    Finally, \Cref{lem:tvdist_after_conditioning} implies 
    \begin{align*}
        \tvdist{g(\Xi) - \Dcal^{k}} &\ge 1 - \pbra{2^{|S|}+1}\cdot 2 \exp\cbra{-\frac{\delta^2 \cdot r}{2}} \\
        &\ge 1 - 4 \exp\cbra{-\frac{\delta^2 \cdot k}{4\cdot(16d\ell/\delta^2)^{2d\ell + 1}}} \\
        &\ge 1 - 4 \exp\cbra{-\pbra{\frac{\delta^2}{16d\ell}}^{4d\ell} \cdot k}. \qedhere
    \end{align*}
\end{proof}

\section*{Acknowledgements}
AO thanks Farzan Byramji for suggestions improving the presentation.
KW thanks David Gosset for helpful discussions motivating the problem.

\bibliographystyle{alpha}
\bibliography{ref}

\end{document}